\RequirePackage{silence}
\documentclass[a4paper,twocolumn,unpublished]{quantumarticle}
\pdfoutput=1
\usepackage[numbers,sort&compress]{natbib}

\usepackage[utf8]{inputenc}
\usepackage[english]{babel}
\usepackage[T1]{fontenc}

\usepackage{amssymb,amsmath,amsthm,mathtools}
\usepackage{etoolbox} 
\usepackage{dsfont} 
\usepackage{appendix}
\usepackage{paralist}
\usepackage{enumitem} 
\setlist{nolistsep} 

\usepackage{algorithm,algorithmicx}
\usepackage[noend]{algpseudocode}

\usepackage{stmaryrd} 
\SetSymbolFont{stmry}{bold}{U}{stmry}{m}{n} 
\usepackage{algpseudocode}
\usepackage[printonlyused,withpage]{acronym}
\makeatletter
\AtBeginDocument{%
 \renewcommand*{\AC@hyperlink}[2]{%
   \begingroup
     \hypersetup{hidelinks}%
     \hyperlink{#1}{#2}%
   \endgroup
 }%
}

\let\originalleft\left
\let\originalright\right
\renewcommand{\left}{\mathopen{}\mathclose\bgroup\originalleft}
\renewcommand{\right}{\aftergroup\egroup\originalright}

\usepackage[table,usenames,dvipsnames]{xcolor}
\usepackage{soul}
\usepackage[colorlinks=true,
			hyperindex,
		    linkcolor = NavyBlue,
		    anchorcolor = hyperrefblue,
		    citecolor = Green,
		    filecolor = hyperrefblue,
		    urlcolor = NavyBlue,
			breaklinks=true]{hyperref}

\usepackage{cleveref}
\crefname{equation}{Eq.}{Eqs.}
\Crefname{equation}{Equation}{Equations}
\crefname{figure}{figure}{figures}
\Crefname{Ffigure}{Figure}{Figures}

\providecommand{\pra}{Phys.\ Rev.\ A}
\providecommand{\prb}{Phys.\ Rev.\ B}

\providecommand{\prl}{Phys.\ Rev.\ Lett.}

\providecommand{\prx}{Phys.\ Rev.\ X}

\newtheorem{corollary}{Corollary}

\newtheorem{theorem}{Theorem}
\newtheorem{lemma}{Lemma}

\newcommand{\e}{\ensuremath\mathrm{e}}
\renewcommand{\i}{\ensuremath\mathrm{i}}

\DeclareMathOperator{\Tr}{Tr}

\renewcommand{\Im}{\operatorname{Im}}

\newcommand{\fro}{\mathrm{F}}

\DeclareMathOperator*{\argmin}{arg\min}
\DeclareMathOperator*{\argmax}{arg\max}

\renewcommand{\L}{\operatorname{L}}
\newcommand{\id}{\mathrm{id}}
\DeclareMathOperator{\U}{U}

\newcommand{\CC}{\mathbb{C}}
\newcommand{\RR}{\mathbb{R}}

\newcommand{\NN}{\mathbb{N}}
\newcommand{\1}{\mathds{1}}
\newcommand{\EE}{\mathbb{E}}

\newcommand{\argdot}{{\,\cdot\,}}
\renewcommand{\vec}[1]{\boldsymbol{#1}}

\DeclarePairedDelimiterX{\abs}[1]{\lvert}{\rvert}{%
  \ifblank{#1}{\argdot}{#1}
}   

\DeclarePairedDelimiterX\norm[1]\lVert\rVert{%
  \ifblank{#1}{\argdot}{#1}
}   

\newcommand{\pnorm}[2][p]{\norm{#2}_{#1}} 

\DeclarePairedDelimiterX{\iiiNorm}[1]{\lvert}{\rvert}{%
  \delimsize\lvert\delimsize\lvert#1\delimsize\rvert\delimsize\rvert%
}

\DeclarePairedDelimiterXPP\snorm[1]{}\lVert\rVert{_\infty}{\ifblank{#1}{\argdot}{#1}}   

\DeclarePairedDelimiterXPP\twonorm[1]{}\lVert\rVert{_2}{\ifblank{#1}{\argdot}{#1}}   

\DeclarePairedDelimiterXPP\onenorm[1]{}\lVert\rVert{_1}{\ifblank{#1}{\argdot}{#1}}   

\DeclarePairedDelimiterXPP\trnorm[1]{}\lVert\rVert{_1}{\ifblank{#1}{\argdot}{#1}}   

\DeclarePairedDelimiterXPP\fnorm[1]{}\lVert\rVert{_{\fro}}{\ifblank{#1}{\argdot}{#1}}   

\DeclarePairedDelimiterXPP\dnorm[1]{}\lVert\rVert{_\diamond}{\ifblank{#1}{\argdot}{#1}}   

\DeclarePairedDelimiterXPP\cbnorm[1]{}\lVert\rVert{_\mathrm{cb}}{\ifblank{#1}{\argdot}{#1}}   
\DeclarePairedDelimiterXPP\ddnorm[1]{}\lVert\rVert{_{\diamond\rightarrow \diamond}}{\ifblank{#1}{\argdot}{#1}}   
\DeclarePairedDelimiterXPP\ssnorm[1]{}\lVert\rVert{_{\infty\rightarrow\infty}}{\ifblank{#1}{\argdot}{#1}}   

\DeclarePairedDelimiterX\Set[1]\{\}{%
  
  #1
}

\DeclarePairedDelimiterX\innerp[2]{\langle}{\rangle}{%
  \ifblank{#1}{\argdot}{#1} , \ifblank{#2}{\argdot}{#2}%
}

\DeclarePairedDelimiter{\bra}{\langle}{\vert}
\DeclarePairedDelimiter{\ket}{\vert}{\rangle}

\DeclarePairedDelimiterX\braket[2]{\langle}{\rangle}%
  {#1\kern0.15ex\delimsize\vert\kern0.15ex\mathopen{}#2}

\DeclarePairedDelimiterX\ketbra[2]{\vert}{\vert}%
  {#1\kern0.15ex\delimsize\rangle\delimsize\langle\kern0.15ex\mathopen{}#2}

\DeclarePairedDelimiterX\sandwich[3]{\langle}{\rangle}%
  {#1\,\delimsize\vert\kern0.15ex\mathopen{}#2\kern0.15ex\delimsize\vert\kern0.15ex\mathopen{}#3}

\DeclarePairedDelimiterX{\av}[1]{\langle}{\rangle}{%
  \ifblank{#1}{\argdot}{#1}
}   

\newcommand{\class}[1]{{\ensuremath{\mathsf{#1}}}}
\newcommand{\NP}{\class{NP}}

\newcommand{\MaxCut}{\class{MaxCut}}

\renewcommand{\S}{S^{\vec x}}

\DeclarePairedDelimiterX{\expvar}[1]{\llbracket}{\rrbracket}{%
  \ifblank{#1}{\argdot}{#1}
}   

\renewcommand{\c}{\vec w}
\newcommand{\ci}{w_i}
\newcommand{\cj}{w_j}
\newcommand{\cscalar}{w}

\newcommand{\om}{relative correlation}
\newcommand{\mg}{m_g}

\newcommand{\sn}[1]{\av{#1}_{s}}
\newcommand{\edis}[1]{\av{#1}_{\vec \theta}}
\newcommand{\edistheta}[1]{\av{#1}_{\vec \theta}}
\newcommand{\edisboth}[1]{\av{#1}_{s,\vec \theta}}
\newcommand{\edisG}[1]{\av{#1}_{G}}

\newcommand{\specwidth}{\nu}
\newcommand{\estat}{\epsilon_\mathrm{stat}}
\newcommand{\esys}{\epsilon_\mathrm{sys}}
\newcommand{\epstot}{\epsilon_{\mathrm{tot}}}
\newcommand{\deriv}{\delta}
\newcommand{\sigakt}{\edistheta{a^2_k}}
\newcommand{\sigalt}{\edistheta{a^2_l}}
\newcommand{\prefac}{\zeta}

\newcommand{\hhu}{%
  Institute for Theoretical Physics, 
  Heinrich Heine University D{\"u}sseldorf,
  Germany%
}

\title{Fast gradient estimation for variational quantum algorithms} 

\author{Lennart Bittel}
\email{lennart.bittel@uni-duesseldorf.de}
\author{Jens Watty}
\email{jens.schneider@uni-duesseldorf.de}
\author[MK]{Martin Kliesch}
\email{science@mkliesch.eu}
\affiliation{\hhu}

\begin{document}
\maketitle 

\hypersetup{%
       pdftitle = {Fast gradient estimation for variational quantum algorithms},
       pdfauthor = {Lennart Bittel, Jens Watty, Martin Kliesch},
       pdfsubject = {Quantum computation},
       pdfkeywords = {Hybrid, quantum, algorithm, NISQ, 
              measurement, derivative, gradient, descent, sample complexity, 
              VQE, variational, eigensolver,
              VQA, algorithms, 
              QAOA, approximate, optimization, 
              generalized, PSR, parameter shift rule, 
              PQC, parametrized ,parameterized, circuit, 
              finite, central, differences, 
              barren, plateau, 
              Hamiltonian, generator, 
              shot noise, variance, 
              statistical, systematic, error, 
              L1, norm, regularization, 
              compressed, compressive, sensing, 
              Bayesian, estimation, prior, 
              MMSE, minimum mean squared%
              }, 
      }

\begin{abstract}
Many optimization methods for training variational quantum algorithms are based on estimating gradients of the cost function. Due to the statistical nature of quantum measurements, this estimation requires many circuit evaluations, which is a crucial bottleneck of the whole approach. 
We propose a new gradient estimation method to mitigate this measurement challenge and reduce the required measurement rounds. Within a Bayesian framework and based on the generalized parameter shift rule, we use prior information about the circuit to find an estimation strategy that minimizes expected statistical and systematic errors simultaneously. We demonstrate that this approach can significantly outperform traditional gradient estimation methods, reducing the required measurement rounds by up to an order of magnitude for a common QAOA setup. 
Our analysis also shows that an estimation via finite differences can outperform the parameter shift rule in terms of gradient accuracy for small and moderate measurement budgets. 
\end{abstract}

\section{Introduction}
It has been demonstrated that quantum devices can outperform classical computers on computational problems specifically tailored to the hardware \cite{Arute2019QuantumSupremacy_short_auth,Zhong20shortened_auth}. 
While this has been an important milestone, the ultimate goal is a \emph{useful quantum advantage}, i.e.\ a similar speedup for a problem with relevant applications. 
The central practical challenge is that only \ac{NISQ} hardware is available for the foreseeable future \cite{Pre18}. 
This restriction means that quantum devices have limited qubit numbers and can only run short quantum circuits, as the quantum computation must be finished before noise effects become too dominant. 
For this reason, great efforts are being made to design quantum algorithms in a \ac{NISQ}-friendly way. 
One central idea in this effort is to trade an increased number of circuit evaluations and additional classical computation for reduced qubit numbers and lower circuit depths. 

One of the leading approaches toward achieving useful quantum advantages is given by \acp{VQA}. 
They address the problem of estimating the ground state energy of a quantum many-body Hamiltonian via a variational optimization, as follows. 
The quantum part of \acp{VQA} is implemented using \acp{PQC}, which are used to prepare the variational quantum states. 
In order to interface it with a classical computer, the energy and the energy gradient w.r.t.\ the variational parameters are typically estimated. 
Then a classical computer, which typically runs some gradient descent-based algorithm, is used to minimize the energy via repeated parameter updates and estimations of the energy functional. 

\begin{figure}[t]
    \centering
    \includegraphics[width=1.04\linewidth]{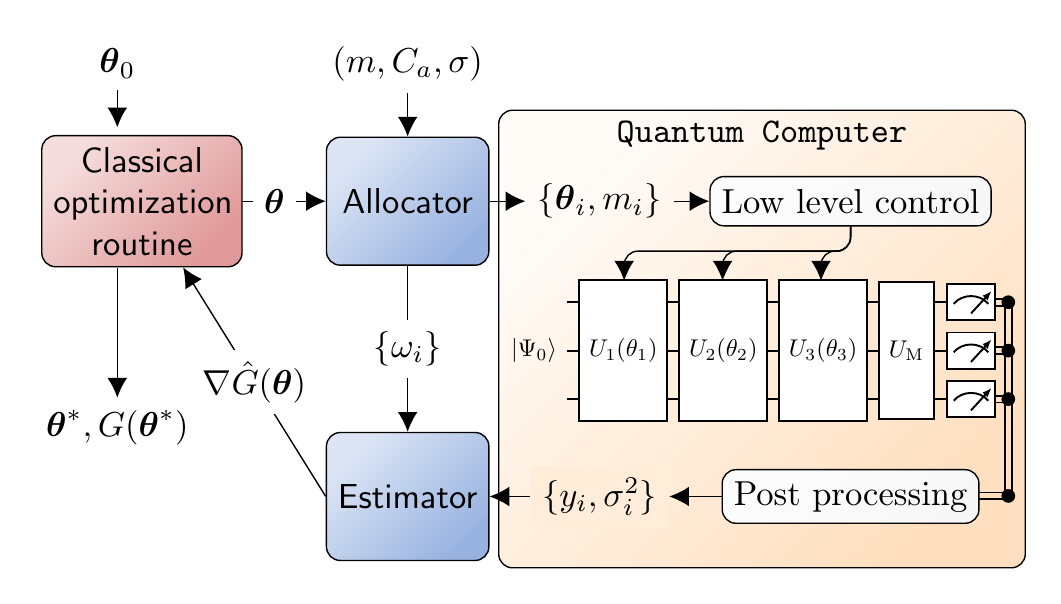}
    \caption{
    Sketch of the gradient estimation for \ac{VQA} optimization routines as shown in \cref{alg:bay}.
    }
    \label{fig:sketch}
\end{figure}

Several challenges occur in this approach. 
First, on the classical computation side, the optimization might reach a barren plateau for the objective function \cite{McClean2018BarrenPlateausIn} or get stuck in local minima \cite{Bittel21TrainingVariationalQuantum}. 
Barren plateaus can sometimes be avoided by using smart initializations for the parameters \cite{Zhou2018QuantumApproximate}. 
Moreover, sophisticated constructions of the quantum circuit family can help to bypass such problematic regions in the parameter space \cite{Grimsley19AnAdaptiveVariational,Grimsley22ADAPT-VQE}. 
Local minima can at least partially be avoided using natural gradients \cite{Wierichs2020AvoidingLocalMinima}. 
Second, the measurement effort of the quantum computer can pose a critical bottleneck for \acp{VQA}. 
The reason is that 
\begin{compactenum}[(i)]
\item many iterations steps are done in the classical optimization, 
\item several partial derivatives are needed for each gradient update step,
\item multiple measurement settings might be needed for the estimation of observables such as local Hamiltonians,
\item quantum measurements are probabilistic, requiring $O(1/\epsilon^2)$ measurement rounds for $\epsilon$ accuracy
\end{compactenum}
and this can add up to a large number of total number of measurements rounds. 
Since quantum measurements are destructive, one also needs to prepare the entire variational state from scratch for each measurement round.

In this work, we develop a new gradient estimation algorithm that balances statistical and systematic errors which achieve a better gradient estimate with fewer measurement rounds than conventional estimators.
Specifically, we first characterize both the statistical and systematic error that arise in the estimation procedure, 
where for the systematic error, we introduce a Bayesian framework using prior information and assumptions about the system to estimate it. 
Then, we develop allocator methods, which for a given measurement budget determine an optimal strategy, namely what and how often we want to measure each circuit configuration, in order to minimize the total error. Finally the estimator returns a gradient estimate based on the measurement outcomes.
A sketch of the procedure is shown in \cref{fig:sketch}. 

The Bayesian approach takes advantage, but also requires prior knowledge about the system. 
We develop strategies to obtain this prior information depending on the circuit depth: 
\begin{compactenum}[(i)]
\item For short circuit, we use experimental or numerical/analytical observations. 
\item For higher depths, we use unitary 2-design properties of random circuits. 
\item In the intermediate regime, we use an interpolation of the two. 
\end{compactenum}

For the analysis in this paper we neglect all error sources arising from imperfect quantum hardware and only focus on noise due to finite measurements (i.e.\ shot noise). 
Additionally, we assume time periodic unitaries, meaning that, without loss of generality, the eigenenergies of the generators can be assumed to be integers. 

Finally, we demonstrate numerically, that the Bayesian approach outperforms previous \ac{PSR} approaches in terms of gradient estimation and \ac{VQA} optimization accuracy.

\subsection{Related work}
There are several approaches to estimate gradients in \acp{VQA}, notably the \ac{PSR} \cite{Li17HybridQuantum-Classical,Mitarai18QuantumCircuitLearning} is able to obtain unbiased gradient estimates for generators with only two distinct eingenvalues. 
Further generalizations were made in Ref.~\cite{Schuld19EvaluatingAnalyticGradients} for a wider class of Hamiltonians. 
This approach often requires ancillary qubits or unitaries generated by commuting generators with two eigenvalues. 
There are also generalizations proposed for non-commuting generators~\cite{Banchi21MeasuringAnalyticGradients} in a stochastic framework. 
These approaches generally require measurement settings not contained in the \ac{VQA}-ansatz, but which can be assumed to be feasible for real hardware. 
There also exists unbiased estimators for arbitrary periodic unitaries~\cite{Wierichs21GeneralParameter-Shift,vidal_calculus_2018}, where all measurements are contained in the ansatz class.

There are also strategies~\cite{Harrow19LowDepthGradient, crooks2019psrGateDecomposition, kottmann2021fermionicPSR} that replace the actual observable underlying the gradient estimation with some surrogate observables. 
Another research direction has been to find efficient estimation schemes for the whole gradient~\cite{Gilyen17OptimizingQuantumOptimization,Cade20StrategiesForSolving} instead of its individual partial derivatives and thus reducing the required measurement resources. 

We approach the gradient estimation differently. 
Namely, we use that \acp{VQA} due to their setup experience some typical behavior,
which can be analyzed in advance and during the \ac{VQA} optimization. 
This allows us to also evaluate the performance of biased estimation strategies, which under very reasonable assumptions on measurement budget, can significantly outperform their unbiased counterparts. 
We use the general framework of periodic parametrized quantum gates \cite{Wierichs21GeneralParameter-Shift} but believe that a similar Bayesian reasoning can also benefit other \ac{VQA} ansatz classes and gradient strategies. 
It should therefore not be regarded as a competitor to existing methods, but as a complementary strategy to further reduce the measurement effort of gradient estimation strategies.


\subsection{Notation}
We use the notation $[n]\coloneqq \{1,\dots,n\}$.
The Pauli matrices are denoted by $X$, $Y$ and $Z$.
An operator $O$ acting on subsystem $j$ of a larger quantum system is denoted by $O_{j}$, e.g.\ $X_1$ is the Pauli-$X$-matrix acting on subsystem $1$.
$\ell_p$-norms including $p = 0$ are denoted by $\pnorm{}$.
We use several symbols that are summarized in \cref{sec:symbols}. 

\section{\texorpdfstring{\Aclp{VQA}}{Variational quantum algorithms}}
In a \ac{VQA} the goal is to find parameters $\vec\theta$ such that a cost function is minimized.
In general, this cost function is given by
\begin{align}\label{eq:VQA}
    G(\vec\theta) \coloneqq \bra{\Psi_0} \left[\prod_{\alpha=1}^L U_\alpha\left(\theta_\alpha\right)\right]^\dagger O \left[\prod_{\alpha=1}^L U_\alpha\left(\theta_\alpha\right)\right] \ket{\Psi_0}\,
\end{align}
where $\ket{\Psi_0}$ is the initial state, $U_\alpha\left(\theta_\alpha\right)=\e^{-\i \theta_\alpha H_\alpha}$ are the unitary gates generated by $H_\alpha$ and $O$ is the observable encoding the optimization problem.
In this work, we are considering unitaries that are $T$-periodic in $\theta_l$ (w.l.o.g. $T=2\pi$), which implies that all eigenvalues of $H_l$ are integers.

Estimating the gradient of $G(\vec\theta)$  w.r.t.\ $\vec\theta$ is an important task, as most optimization algorithms are gradient descent based and thus require an efficient approximation of the gradient.
Our strategy estimates the gradient by the functions partial derivatives. For this it is convenient to define the cost function at a point shifted by a value of $x$ in the parameter $\theta_l$
\begin{align}
\begin{aligned}
    F_l(x)&\coloneqq G(\vec \theta+ x\,\vec e_l )\\
  &=\bra{\Psi'} U_l^\dagger(x)O' U_l(x)\ket{\Psi'}\,,
\end{aligned}
\end{align}
where $\vec e_l$ is the $l$-th canonical basis vector and the other layers are absorbed into the observable as $O'$ and the initial state as $\ket{\Psi'}$.
Furthermore, we used that $U_l(\theta_l+x)=U_l(\theta_l)U_l(x)$.
For later reference, the modified state and observable are
\begin{align}
\begin{aligned}
        \ket{\Psi'} &= \left[\prod_{\alpha=1}^{l-1} U_\alpha\left(\theta_\alpha\right)\right]\ket{\Psi_0}\, \text{, and}\\
    O' &= \left[\prod_{\alpha=l}^L U_\alpha\left(\theta_\alpha\right)\right]^\dagger O \left[\prod_{\alpha=l}^L U_\alpha\left(\theta_\alpha\right)\right].
\end{aligned}
\end{align}

The evaluation at point $\vec \theta$, is therefore just $G(\vec \theta)=F_l(0)$ and $\frac{\mathrm{d} G(\vec \theta)}{\mathrm{d}\theta_l}=F_l'(0)$. 
We will henceforth focus only on the estimation of a single partial derivative w.r.t.\ a parameter $\theta_l$. 
In the interest of readability we write 
$U$, and $F$ instead of $U_l$, and $F_l$, as well as $\ket{\Psi}$ and $O$ instead of $\ket{\Psi'}$ and $O'$.

In this restricted view, we are now going to examine the structure of the function $F(x)$ more closely.
The parametrized unitary that defines this function has the form
\begin{align}
	U(\theta)=\e^{-\i \theta H}=\sum_{i=1}^{n_{\lambda}} P_i \e^{-\i \lambda_i \theta}\,,
\end{align}
where $H$ is a Hermitian generator and the $P_i$ are the projectors onto the eigenspaces corresponding to the eigenvalues $\{\lambda_1,\dots,\lambda_{n_\lambda}\}$ of $H$ in ascending order.

Using this notation, we obtain
\begin{equation}\label{eq:costfun_onelayer}
\begin{aligned}
    F(x)&=\bra{\Psi} U^\dagger(x)O U(x)\ket{\Psi}\,\\
  &=\sum_{i,j=1}^{n_{\lambda}} \e^{\i(\lambda_i-\lambda_j)x}\bra{\Psi} P_i O P_j\ket{\Psi}\,,
\end{aligned}
\end{equation}
where each $c_{ij}\coloneqq \bra{\Psi}P_iOP_j\ket{\Psi}$ is just a scalar. 
This lets us rewrite the function as
\begin{align}
  \label{eq:sincosform}
  F(x) &=\sum_{i,j=1}^{n_{\lambda}} \e^{\i(\lambda_i-\lambda_j)x}c_{ij}
  =\sum_{k=1}^{n_\mu} c_k\e^{\i\mu_k x}+ c_k^*\e^{-\i\mu_k x}\nonumber\\
  &=\sum_{k=1}^{n_\mu} \bigl( a_k\sin(\mu_kx)+b_k\cos(\mu_kx) \bigr) \,,
\end{align}
where $\mu_k\in\{\left|\lambda_i-\lambda_j\right|\}$ are all possible eigenvalue differences of the generator $H$, 
\begin{align}
c_k=\sum_{i,j:\lambda_i-\lambda_j=\mu_k} c_{ij}
\end{align}
and $\vec c=\frac{\vec b +\i \vec a}{2}$ with $\vec a, \vec b \in \RR^{n_\mu}$ are Fourier coefficients.
For the total number of frequencies, it follows $n_{\mu} = \abs{\{\mu_k\}} \leq \binom{n_\lambda}{2}$.
The derivative at $x=0$ is
\begin{align}
    \label{eq:real_deriv}
    \deriv\coloneqq F'(0)=\sum_{k=1}^{n_\mu} \mu_k a_k\,.
\end{align}

For generators with two eigenvalues, where we can set w.l.o.g.\ $\mu_k \in \{0,\nu\}$, it has been shown that an unbiased estimate for the partial derivative can be obtained via
\begin{align}\label{eq:psr}
	\deriv=\nu\frac{F(\frac{\pi}{2\nu})-F(-\frac{\pi}{2\nu})}{2} \,,
\end{align}
which is known as the \ac{PSR} \cite{Schuld19EvaluatingAnalyticGradients}. 

In essence, we are going to generalize this method.
A helpful tool for this task is the antisymmetric projection
\begin{align}
    \label{eq:sinform}
    f(x)=\frac{F(x)-F(-x)}{2}= \sum_{k=1}^{n_\mu}a_k\sin(\mu_k x)\,.
\end{align}
We are only considering symmetric measurement schemes as symmetrizing an estimation method will not make the prediction worse~\cite{Wierichs21GeneralParameter-Shift}. 
As such, we refer only to the positive measurement positions $x$ of $f(x)$, 
knowing that estimates of $F(+x)$ and $F(-x)$ are required to determine it. 
We will also omit the $\mu=0$ frequency, since it does not affect the derivative. Additionally, $\nu\coloneqq\left\|\vec \mu\right\|_\infty$ refers to the spectral width of the generator, meaning $\vec \mu\subset[ \nu ]$, since the periodicity of $U(x)$ implies that the entries of $\vec \mu$ are positive integers.

\newcommand{\allocator}{allocator}
\section{Gradient estimation approach}
The \emph{allocator} decides the measurement resource allocation and how to generate the estimate. In particular, we use a symmetric linear estimator of the derivative which for a set of measurement positions $\vec x\in [0,\pi)^{n_x}$ and number of measurements for each position $\vec m\in \NN^{n_x}$ returns a gradient estimate
\begin{align}
    \hat{\deriv}=\sum_{i=1}^{n_{x}} \ci y_i \label{eq:linest}\,,
\end{align}
where $y_i$ are the empirical estimates of $f(x_i)$ using $m_i$ measurement rounds. Since each $x_i$ requires $2$ measurement settings, the total number of settings is $2n_x$.
In the following we develop strategies to find optimal $\vec x,\vec m$ and $\c$.

The error 
\begin{align}
    \epstot\coloneqq \hat{\deriv}-\deriv
\end{align}
between our estimator guess $\hat\deriv$ and the true derivative $\deriv$
is an important metric that we are going to use as a figure of merit for choosing our estimation parameters.
In practice, imperfections of current quantum hardware and the lack of quantum error correction will cause a significant noise level when evaluating the cost function on the quantum device.
However, even on a perfect device, the measurement process introduces a shot noise error as the estimates are determined by sampling from the underlying multinomial distribution.
We denote the expectation value over the shot noise by $\sn{\argdot}$.

The expected error under shot noise can then be written as
\begin{align}
    \sn{\epstot}&=\sum_{i=1}^{n_x} \ci \sn{y_i}-\deriv\\
    &=\sum_{i=1}^{n_x} \ci\left(\sum_{k=1}^{n_\mu}a_k\sin(\mu_kx_i)\right)-\sum_{k=1}^{n_\mu} \mu_ka_k\\
    &\eqqcolon (\S\c-\vec\mu)^T \vec a\, ,
\end{align} 
where we have used the definition (\cref{eq:real_deriv}) of $\deriv$, the unbiased nature of the estimate $\sn{y_i}=f(x_i)$ and \cref{eq:sinform} for $f$. We also defined the matrix $\S$ with entries $\S_{ki}\coloneqq \sin(\mu_k x_i)$. 

For the mean squared error this means
\begin{align}
    \sn{\epstot^2}&=\av*{\Big(\sum_{i=1}^{n_x} \ci y_i-\deriv\Big)^2}_s \nonumber\\
    &=\left[(\S\c-\vec\mu)^T \vec a\right]^2+\sum_i \ci^2\left(\sn{y_i^2}-\sn{y_i}^2\right)\nonumber\\
    &\eqqcolon \esys^2+\sn{\estat^2} \, ,
\end{align} 
where the first term describes the systematic error resulting from the method not accurately determining the derivative even for exact measurements and the second term describes the statistical error arising from measurement shot noise. 

\subsection{Estimating the statistical error}
For the statistical error, each term $\sn{y_i^2}-\sn{y_i}^2$ is the variance for 
the measurement position $x_i$ resulting from shot-noise errors. 
If the single shot variance at position $x_i$ is $\sigma_i^2$, we find the expression \begin{align}\sn{y_i^2}-\sn{y_i}^2=\frac{\sigma_i^2}{m_i}\,
\end{align}
 where $m_i$ is the number of measurement rounds performed for $x_i$.
For a fixed measurement budget $m=\sum_{i=1}^{n_x} m_i$, the optimal measurement allocation is given by
\begin{align}
    \sn{\estat^2}= \min_{\vec m: \|\vec m\|_1=m}\sum_{i=1}^{n_x} \ci^2\frac{\sigma_i^2}{m_i}&=\frac{(\sum_{i=1}^{n_x}|\ci|\sigma_i)^2}{m}\,,\nonumber
\end{align}
with $m_i=m\frac{|\ci| \sigma_i}{\sum_{j=1}^{n_x}|\cj|\sigma_j}$. If one assumes constant shot noise $\sigma_i\equiv \sigma$ regardless of the measurement position which we will hence force do, this simplifies to 
\begin{align}\label{eq:estat_uni}
    \sn{\estat^2}=\frac{\sigma^2}{m}\norm{\c}_1^2\quad \textrm{with}\quad  m_i=\frac{|\ci|}{\norm{\c}_1}m\,.
\end{align}
While $\sigma_i$ or $\sigma$ are a priori unknown, it is generally possible to give a rough estimate of $\sigma$ beforehand and to determine an estimate of $\sigma_i$ after only a few measurement rounds are performed. 
For this reason we assume that $\sigma$ is a known quantity in the following sections.

\subsection{Estimating the systematic error through a Bayesian approach}
Determining the systematic error is more challenging because it depends explicitly on the Fourier coefficients $\vec a$ which are not known. 
For our analysis we assume that the estimator will be used for an ensemble of multiple different positions $\vec \theta$, as one expects to occur during a full gradient descent optimization routine. So instead of a particular instance we want to find a strategy where the average total error over the entire ensemble is minimized.
The benefit of this approach is that the average only requires knowledge of the general behavior of the Fourier coefficients, not specific values of the particular realization.

Formally, we assume that a distribution $\mathcal{D}_{\vec \theta}$ over the relevant positions $\vec \theta$ induces a distribution of the Fourier coefficients $\vec a$.
One natural distribution $\mathcal{D}_{\vec \theta}$ is the uniform distribution over all parameter points $\vec\theta\in[-\pi,\pi)^L$, which can be motivated as a model for the case of random initialization $\vec \theta_0$.   
For an expectation value over the distribution $\mathcal{D}_{\vec \theta}$, we write $\edis{\argdot}$.
In this framework, taking the expectation value for $\mathcal{D}_{\vec \theta}$ yields
\begin{align}
\begin{aligned}
        \edistheta{\esys^2}&=\av*{\left[(\S\c-\vec \mu)^T \vec a\right]^2}_{\vec \theta}\\
    &=(\S\c-\vec \mu)^T \edistheta{\vec a\vec a^T}(\S\c-\vec \mu)\\
    \label{eq:toterr_dis}
    &\eqqcolon(\S\c-\vec \mu)^T C_a(\S\c-\vec \mu)\,,
\end{aligned}
\end{align}
meaning that the estimation of the expected squared systematic error requires knowledge of the second moment matrix 
\begin{equation}\label{eq:Ca}
   C_a\coloneqq\edis{\vec a\vec a^T}\in \RR^{n_{\mu}\times n_{\mu}}\,.
\end{equation} 
In the following, we derive properties of the of $C_a$ assuming the uniform distribution over $\vec \theta$. 

First, we note that for a shift $\vec\theta \rightarrow \vec\theta +\vec e_l z$ in the layer~$l$, the complex Fourier coefficients transform as $c_k\rightarrow c_k\e^{\i \mu_k z}$, see \cref{eq:sincosform}. 
As $\mathcal{D}_{\vec\theta}$ is invariant under such a shift, we have 
\begin{align*}
    \edistheta{c_k}
    =\frac{1}{2\pi}\int_{-\pi}^\pi \edistheta{c_k}\mathrm d z
    =\frac{1}{2\pi}\int_{-\pi}^\pi \edistheta{c_k \e^{\i\mu_k z}}\mathrm d z
    =0
\end{align*}
implying
that $c_k$ is centered around $c_k=0$ in expectation over $\vec \theta$. 
Similarly for the second moment, we compute
\begin{align*}
    \edistheta{c_kc_p}&=\frac{1}{2\pi}\int_{-\pi}^\pi \edis{c_k c_p}\e^{\i(\mu_k+\mu_p) z}\mathrm d z=0\,,\\
    \edistheta{c_k^*c_p}&=\frac{1}{2\pi}\int_{-\pi}^\pi \edis{c_k^* c_p}\e^{\i(\mu_k-\mu_p) z}\mathrm d z=\edis{|c_k|^2}\delta_{kp}\,,
\end{align*}
where $\delta_{kp}$ 
is 
the Kronecker delta.
From the Fourier expansion~\cref{eq:sincosform} and 
$a_k=2\Im(c_k)$, 
it follows that
\begin{align}
\begin{aligned}
    \edistheta{a_k}&=0\\
    \edistheta{a_ka_p}
    &=2\,\delta_{kp}\edis{|c_k|^2}\, ,
\end{aligned}
\end{align}
meaning that $C_a$ 
is 
a diagonal matrix with entries 
$\sigakt=2\edis{\abs{c_k}^2}$ 
given by
the expected squares of the Fourier coefficients.

What remains is determining $\sigakt$. It is worth pointing out that while underestimating $\sigakt$ can lead to suboptimal results, even significantly overestimating the amplitudes will still outperform methods, where no prior assumptions are made, meaning rough estimates of $\sigakt$ are already sufficient for good performance.
One way of estimating them is to use already existing empirical measurement data from previous optimization rounds or initial calibration. 
The coefficients can be estimated using a Fourier fit.
Another strategy involves numerically simulating smaller system sizes and extrapolating to the actual size used in the \ac{VQA}.

If the applied unitaries in the \ac{VQA} are known, $\sigakt$ can sometimes be derived theoretically.
For instance in \cref{ap:l1_est}, we derive analytically exact results for a \ac{VQA} with a single layer ($L=1$).
For a small constant circuit depth, $\sigakt$ can be computed efficiently, using Monte-Carlo sampling algorithms, even for large system sizes.
This is convenient, as the case for deep circuits, under certain assumptions, $\sigakt$ can be approximated again using only the spectral composition of the generator. This is shown in the following.

\subsubsection*{Ergodic limit -- barren plateaus}
\ac{VQA} optimization routines have to overcome a general phenomenon known as \emph{barren plateaus}.
This term describes the tendency of a gradient in \acp{VQA} to be exponentially suppressed in the system size with increasing circuit depth and for almost all parameters $\vec \theta$. 
This phenomenon has been extensively studied and while mitigation techniques are proposed~\cite{Zhou2018QuantumApproximate,Grimsley19AnAdaptiveVariational,Grimsley22ADAPT-VQE}, 
it appears to be unavoidable, at least in the general setup. 

For the rigorous analysis of barren plateaus, it is beneficial to use the language of unitary $t$-designs. 
Effectively, with increasing circuit depth, the overall applied gate will appear more and more like a Haar-random unitary, with respect to which the derivative is suppressed by the Hilbert space dimension. 
This assumes that the underlying generators describe a universal gate set.
For this effect to occur, we do not need convergence to the Haar measure but convergence to a unitary $2$-design, which is significantly quicker~\cite{harrow_random_2009,brandao_local_2016}. 
For our purposes, such approximate unitary $2$-designs are sufficient since $\sigakt$ can be expressed as a polynomial in $U,U^\dagger$ of degree $(2,2)$. 
If this condition is met, we can replace the expectation value over all angles by the expectation value over all unitaries. 
Hence, this condition can be summarized by the \emph{ergodic assumption}
\begin{align}\label{eq:ergodic_argument}
    \edistheta{\Gamma(U(\vec \theta))}
    = 
    \int \Gamma(U)\mathrm{d}U\, ,
\end{align}
which holds for any polynomial $\Gamma(U)$ of degree at most $(2,2)$ and where the integral is taken w.r.t.\ the Haar probability measure on the unitary group. 

Under this assumption we derive in \cref{ap:2design} that 
\begin{align}\label{eq:sigakt_2design}
    \sigakt&=\xi_d\frac{\sigma_O^2}d\sum_{i\geq j:\mu_k=\lambda_i-\lambda_j}\Tr[P_i/d]\Tr[P_j/d]
\end{align}
with $\sigma_O^2\coloneqq \Tr[O^2/d]-\Tr[O/d]^2$, where $d$ is the Hilbert space dimension and $\xi_d$ is a constant close to $1$ that depends only on $d$. 
$\Tr[P_i/d]$ is the relative multiplicity of the eigenvalue $\lambda_i$.  
Notably, the factor of $\frac{1}{d}$ shows the exponential suppression of the derivative in the system size, meaning that gate sets drawn from a 2-design experience barren plateaus. 
For $L\rightarrow\infty$ this result confirms the assumption that the relative frequency of an eigenvalue difference $\mu_k$ in the spectrum of the generator determines the expected size of its respective Fourier coefficient.
We will analyze the strengths and limitations of this approach with an example in \cref{sec:ca}.

\section{Allocation methods}
In this section, we derive several allocation methods for a given measurement budget. 
A Python implementation of these methods is available on GitHub \cite{Bittel22GitHub}. 

The estimation algorithm requires the values $\c,\vec m$ and $\vec x$. 
We have already seen that making assumptions about the ensemble of configurations allows us to estimate the error using the second-moment matrix $C_a$ from \cref{eq:Ca} 
and an a priori shot noise estimate $\sigma^2$. 
In the following, 
we devise explicit measurement procedures by making use of the knowledge of these quantities. 
In \cref{sec:full}, we show that using convex optimization procedures, one can derive an optimal measurement strategy, which we call \ac{BLGE}. 

In \cref{sec:unbiased}, we then consider the case, when the number of total measurements goes to infinity. 
We call this method \ac{ULGE}, as it does not require access to the estimates of $C_a$ and $\sigma^2$ and yields an estimate without systematic error.
\ac{ULGE} is an equivalent formulation of a known method in literature~\cite{Wierichs21GeneralParameter-Shift}. 
In \cref{sec:PSR}, we show strong similarity between \ac{ULGE} and another popular generalized \ac{PSR} found in the literature~\cite{Schuld19EvaluatingAnalyticGradients}.
Finally, in \cref{sec:single}, we restrict the number of measurement positions $n_x$ to 1 and derive a strategy that is optimal under this constraint.
We call this method \ac{SLGE}. 
We note that this solution coincides with the result obtained for the non-restricted problem if only very few total measurements are available for the gradient estimation. 
\newcommand{\slvar}{\kappa}
\subsection{\texorpdfstring{\acl{BLGE}}{Bayesian linear gradient estimator}}
\label{sec:full}
For the linear estimator of a partial derivative 
we want to minimize the expected squared error $\edisboth{\epstot^2}$ by finding suitable measurement positions $\vec x\in [0,\pi)^{n_x}$ with weights $\c\in\RR^{n_x}$ of our linear estimator, i.e.\ we wish to find the optimal solution of 
\begin{align}\label{eq:primalop}
(\c^*, \vec x^*) = \argmin_{\c, \vec x}\left\{\edisboth{\epstot^2}(\c, \vec x)\right\}
\end{align}
with
\begin{align*}
    \edisboth{\epstot^2}
    &=(\S\c-\vec \mu)^T C_a(\S\c-\vec \mu)+\frac{\sigma^2}{m}\onenorm*{\c}^2\\
    &=\sum_{k=1}^{n_\mu} \sigakt \left(\sum_{i=1}^{n_x} \ci \sin(\mu_k x_i)\right)^2+\frac{\sigma^2}{m}\onenorm*{\c}^2
\end{align*}
 from \cref{eq:toterr_dis} and \cref{eq:estat_uni}. 
Since the cost function is non-convex in $x$, a direct approach may not reach the optimal solution. 
In \cref{sec:optimization} we show that instead, we can be perform an equivalent maximization problem which arises as and effective dual problem after adding certain constraints:
\begin{align}\label{eq:dualop}
    \edisboth{\epstot^2}^*=\max_{\vec \slvar\in \RR^{n_\mu}} g(\vec \slvar)
    \end{align}
with
\begin{align*}
g(\vec \slvar)=\sum_{k=1}^{n_\mu}\left(2\slvar_k \mu_k-  \frac{\slvar_k^2}{\sigakt}\right)-\frac{m}{\sigma^2} \norm[\Big]{\sum_{k=1}^{n_\mu} \slvar_k \sin(\mu_k (\argdot))}_{\infty}^2\nonumber
\end{align*}
where the last term refers to the $L^\infty$-space norm. 
Each dual variable
$\{\slvar_k\}$ 
has the interpretation as the systematic error w.r.t.\ only one frequency component
\begin{align}\label{eq:err_sys_p}
    \edis{\epsilon^2_{\mathrm{sys},k}}=\left(\left(\S\c\right)_k - \mu_k\right)^2 \sigakt=\frac{\slvar_k^2}{\sigakt}\,.
\end{align}
Due to complementary slackness between the primal and the dual problem, the global maxima positions of the function 
\begin{align}
    \rho_{\vec{\slvar}}(x)=\abs[\Big]{\sum_{k=1}^{n_\mu} \slvar_k \sin(\mu_k x)}
    \label{eq:lam_extrema}
\end{align} 
are the set of optimal measurement positions  $\vec x^*$. 
Having determined those, we can proceed determining the weights $\c^*$ by solving the convex problem~\cref{eq:primalop} with the fixed positions $\vec x^*$
and obtain the measurement budget $\vec m$ via \cref{eq:estat_uni}. We also allow for basic post-processing, where after the measurements are performed, we obtain updated weights $\c$  by replacing the statistical error in \cref{eq:primalop} by one using the empirically determined shot-noise variances. 
\begin{align}\label{eq:estst_post}
    \estat^2=\sum_i \ci^2\sigma_{\mathrm{emp},i}^2
\end{align}
where $\vec \sigma_{\mathrm{emp}}^2$ refers to the empirical estimate of the variance.
All the steps are shown in \cref{alg:bay}, also including the two other methods we consider in the following section.

\begin{figure}
\begin{algorithm}[H]
\caption{Gradient estimation\label{alg:bay}}

\begin{algorithmic}
\State 
{\small Allocator and estimator procedure for the three outlined strategies. $\mathcal{M}_{m}(x)$ refers to performing physical measurement at position $x$ using $m$ measurement rounds and estimating the expectation value and variance.}
\Procedure{Allocator}{$\vec \mu$,$C_a$,$\sigma$,$m$}
\If{Bayesian}
    \State $\vec \slvar^*\leftarrow\argmax[g(\vec \slvar)]$ \Comment{\eqref{eq:dualop}} 
    \State $\vec x^*\leftarrow \argmax[\rho_{\vec \slvar^*}(\vec x)]$ \Comment{\eqref{eq:lam_extrema}}
    \State $\vec w^*\leftarrow \argmin\left[\edisboth{\epstot^2}(\c,\vec x^*)\right]$ \Comment{\eqref{eq:primalop}}
\EndIf
\If{Unbiased}
    \State $\vec x^*\leftarrow \frac{\pi}{\specwidth}([\nu]+\frac{1}{2})$
    \State $\vec w^*\leftarrow\frac{(-1)^i}{2\specwidth\sin^2\left(\frac{\pi}{2\specwidth}\left([\nu] + \frac{1}{2}\right)\right)}$ \Comment{\eqref{eq:ubiasedweight}}
\EndIf
\If{Single}
    \State $x^*\leftarrow \argmin\left[\edisboth{\epstot^2}(x)\right]$ \Comment{\eqref{eq:singleop}}
    \State $w^*\leftarrow w(x^*)$
\EndIf
\State $\vec m\leftarrow \mathrm{ROUND}\left( m\frac{\vec w}{2\norm{\vec w}_1}\right)$
\State \Return $\vec x^*, \c^*, \vec m$
\EndProcedure

\Procedure{Estimator}{$\vec x,\c, \vec m$}
\For{$i \in \{1,\dots, n_x\}$}
    \State $y_{(\mathrm{emp},i,+)}$, $\sigma_{(\mathrm{emp},i,-)}^2\leftarrow\mathcal{M}_{m_i}(x_i)$
    \State $y_{(\mathrm{emp},i,-)}$, $\sigma_{(\mathrm{emp},i,-)}^2\leftarrow\mathcal{M}_{m_i}(-x_i)$
    \State $y_{(\mathrm{emp},i)}\leftarrow\frac{y_{(\mathrm{emp},i,+)}-y_{(\mathrm{emp},i,-)}}{2}$
    \State $\sigma^2_{(\mathrm{emp},i)}\leftarrow\frac{\sigma_{(\mathrm{emp},i,+)}^2+\sigma_{(\mathrm{emp},i,-)}^2}{2}$
\EndFor
\If {postprocess}
\State $\c^*\leftarrow \argmin\left[\edistheta{\epstot^2}(\c,\vec x^*)|_{\vec\sigma_\mathrm{emp}^2}\right]$ \Comment{\eqref{eq:estst_post}}
\EndIf
\State $\hat{\deriv}\leftarrow \vec y_\mathrm{emp}^T\c^*$;
\State \Return $\hat{\deriv}$
\EndProcedure
\end{algorithmic}
\end{algorithm}
\end{figure}

\begin{figure}[ht]
    \centering
    \includegraphics[width=0.50\linewidth]{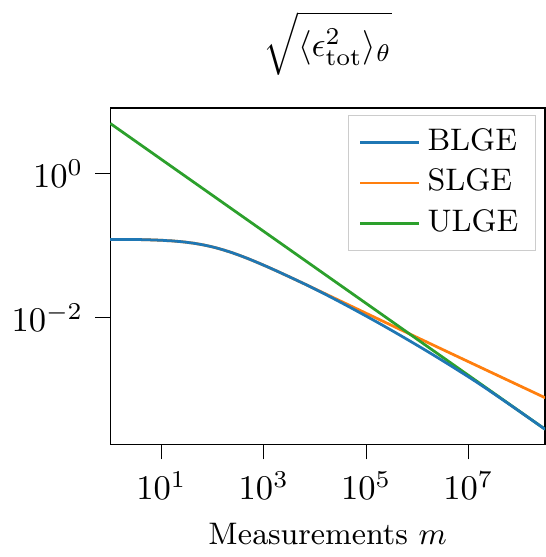}
    \includegraphics[width=0.484\linewidth]{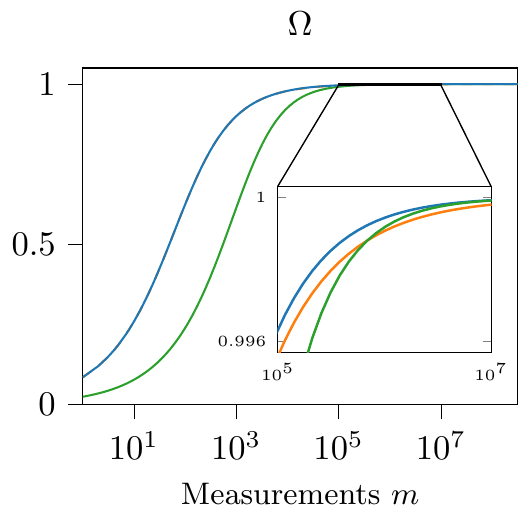}\\
    \vspace{10pt} 
    \includegraphics[width=1\linewidth]{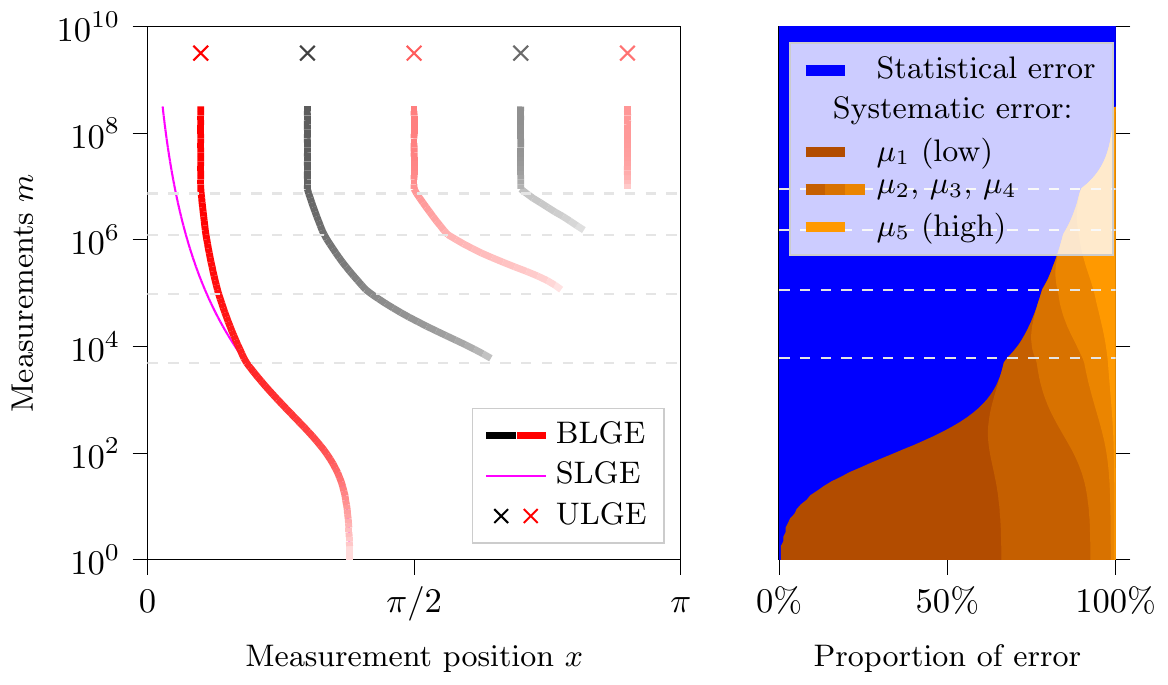}
    \caption{
    \textbf{Top left:} The theoretical mean error $\edisboth{\epstot^2}$ for different measurement budgets for all three methods.
    \textbf{Top right:} The theoretical \om\ $\Omega$ for all methods.
    \textbf{Left:} The measurement positions chosen by the different methods for given total measurement budgets. 
    The color of the lines indicates the sign of the coefficient $\ci$ applied to the value (red - positive; black - negative); the fainter the color, the lower the absolute value of the coefficient $\ci$. 
    \ac{ULGE} is only shown once, since its allocation is independent of $m$. 
    \textbf{Right:} The proportion of the total error that is due to the statistical error (blue) or the systematic error (shades of orange) in the case of the \ac{BLGE}.
    Used values throughout: $\vec \mu = (1,2,3,4,5)$, $\sigma^2=1$, $\edis{a_k^2}=0.1\times10^{-\mu_k}$.
    \label{fig:measpos}
    }
\end{figure}
We also show that the number of measurement settings is restricted by the number of frequency differences in the generator spectrum.
\begin{theorem}\label{thm:finite_positions}
For every optimization problem as defined in \cref{eq:primalop}, there exists an optimal \ac{BLGE} strategy that requires at most $n_x=n_\mu$ unique (positive) measurement positions. (Meaning $2n_\mu$ positions in total)
\end{theorem}
In \cref{ap:proof_of_sparse} we show that this follows directly from the existence of a sparse solution for an $\ell_1$-norm minimization with a linear constraint.

As a quality measure, we define the \emph{\om}\ $\Omega$ as
\begin{align}
    \label{eq:omega}
    \Omega^2\coloneqq\frac{\edisboth{\hat \delta \delta }^2}{\edistheta{\delta^2\vphantom{\hat \delta^2}}\edisboth{{\hat \delta^2}}}\,,
\end{align}
which is the covariance between the real and the estimated derivative scaled onto the interval $\Omega \in [0,1]$. 
Using the quadratic inequality we can find a relationship between the error and the \om
\begin{align}
\begin{aligned}
    \edisboth{\epstot^2}&=\edisboth{\delta^2\vphantom{\hat \delta^2}}+\edisboth{\hat \delta^2}-2\edisboth{\hat \delta \delta }\\
        &\geq\edisboth{\delta^2\vphantom{\hat \delta^2}}-\frac{\edisboth{\hat \delta \delta }^2}{\edisboth{\hat \delta^2}}\\
    \label{eq:eps_to_omega}
    &=\edisboth{\delta^2\vphantom{\hat \delta^2}}\left(1-\Omega^2\right)
\end{aligned}
\end{align}
where the lower bound is saturated by an optimal rescaling of $\c$.
This allows us to think of $\Omega$ as an expected relative accuracy.

\Cref{fig:measpos} shows the result of this method for $n_\mu=5$ and prior estimates $(\sigakt,\sigma^2)$ inspired by the model \ac{VQA} defined in \cref{sec:results}, where small frequencies contribute significantly more to the overall gradient than large frequencies.
We see that in the beginning, the error (blue line, top left) plateaus  and only starts dropping once more than $10^3$ measurement rounds are performed. This occurs as the method returns a very small guess to avoid being wrong, meaning the error is just the expected size of the derivative. For large measurement budgets, we have the expected $\edisboth{\epstot^2}\propto \frac{1}{m}$ behavior. Similarly, for the  \om\ plot (top right), $\Omega$ is increasing from a value close to $\Omega=0$ for very few measurement rounds to $\Omega=1$ for $m\sim 10^4$ measurement rounds.

The bottom left plot show the measurement positions used in the \ac{BLGE}.
The $y$-axis denotes the measurement budget that the allocation method has available.
We see that $n_\mu$ many positions only occur for very large measurement budgets ($10^7$), while only a single position ($2$ when also considering the negative measurement position) is returned for $m\leq10^4$ measurements. For increasing $m$, the measurement positions move further to the left with new ones being added on the right. 

On the bottom right plot, the composition of the error into statistical error (blue) as well as the different systenatic errors based on the different frequency components (shades of orange), as defined in~\cref{eq:err_sys_p}, are shown for different measurement budgets $m$.
For few measurements the largest contribution to the systematic error error dominates with most of the error coming from the Fourier coefficient $a_1$, as this is the most significant component of $f$. 
As $m$ increases, the statistical error starts becomes dominant while the systematic error, first the small frequencies, vanishes. This shows that as one might expect, the estimator becomes more unbiased, as the measurement budget increases.

In the next sections we analyze the behavior in the two limits, the central differences behavior in small $m$ and the unbiased equidistant measurement strategy for large $m$.

\subsection{\texorpdfstring{\Acl{ULGE}}{Unbiased linear gradient estimator}}
\label{sec:unbiased}
If we let $m\rightarrow\infty$, the method will use all resources to make the systematic error vanish exactly, meaning that 
the remaining minimization of the statistical error in~\cref{eq:primalop} simplifies to the convex optimization problem
\begin{align}
    \label{eq:conv_opt_unb}
	 \edisboth{\epstot^2}=\frac{\sigma^2}{m}\min_{\c:\, \S\c=\vec \mu} \onenorm*{\c}^2\,,
\end{align}
where we recall $\vec \mu\subset[\nu]$ with $\nu$ the spectral width. This method is equivalent to one without a systematic error, which has been studied before in more depth~\cite{Wierichs21GeneralParameter-Shift}. Its behavior can be summarized in the following theorem.
\begin{theorem}
For constant shot noise variance $\sigma^2$, any unbiased gradient estimation method given by \cref{eq:conv_opt_unb} has an error
$\edisboth{\epstot^2}\geq\frac{\sigma^2\specwidth^2}{m}$ and $\Omega^2\leq\frac{\edis{\delta^2}}{\edis{\delta^2}+ \frac{\specwidth^2\sigma^2}{m}}$, where tightness can be achieved with at most $n_x=n_\mu$ measurement positions.
\end{theorem}
\begin{proof}
As the method should be unbiased for all functions anti-symmetric function $f(x)$ with the allowed frequencies, we can choose $f(x)=\sin(\nu x)$ to find an upper-bound.: 
\begin{align}
    \delta=\nu=\sum_i \ci \sn{y_i}=\sum_i \ci f(x_i)\,.
\end{align} 
using $|f(x_i)|\leq 1$ and the triangle equality it follows that $\nu \leq \|\c\|_1$.
This means that the error for $m$ measurements is
\begin{align}
    \edisboth{\epstot^2}=\frac{\sigma^2}{m}\onenorm*{\c}^2 \geq
        \frac{\sigma^2\specwidth^2}{m}\,.
\end{align}
Since there is no systematic error we also obtain
\begin{align}
\begin{aligned}
    \Omega^2
    &=\frac{\edis{\delta^2}}{\edis{\delta^2}+ \frac{\specwidth^2\sigma^2}{m}}\\
    &\overset{m\rightarrow 0}{=} \frac{m\edis{\delta^2}}{\sigma^2}\times \frac{1}{\specwidth^2}+O(m^2)\,.
\end{aligned}
\end{align}
What remains to show is that there exists a closed form solution for \cref{eq:conv_opt_unb} which reaches this bound. For this we choose the measurement positions $x_i=\frac{\pi}{\specwidth}(i+\frac{1}{2})$ with $i\in\{0,\dots,\specwidth-1\}$ yielding
\begin{align}
	S_{ik}=\sin\left(\frac{\pi}{\specwidth}\left(i+\frac{1}{2}\right)\mu_k\right)\,.
\end{align}
This describes the discrete sine transform (DST-II), which by inversion yields $\c=S^{-1}\vec \mu$ 
\begin{align}\label{eq:ubiasedweight}
    \ci
       &= \frac{(-1)^i}{2\specwidth\sin^2\left(\frac{\pi}{2\specwidth}\left(i + \frac{1}{2}\right)\right)}\,.
\end{align}
We note that these coefficients 
were
already derived in literature~\cite{Wierichs21GeneralParameter-Shift} using Dirichlet kernels. 
One can verify that $\|\c\|_{1}=\specwidth$, i.e.\ that our choice of coefficients saturates the lower bound we derived.
We note that while this closed form solution requires $\specwidth$ many measurement positions, since it is in the limit of the general Bayesian method, there always exists a strategy with at most $n_{\mu}$ many positions as shown in \cref{thm:finite_positions}.
\end{proof}
As can be seen in \cref{fig:measpos}, especially for small $m$, \ac{BLGE} performs significantly better than \ac{ULGE} for both the expected error and $\Omega$. This is because \ac{ULGE}, in order to be unbiased, is very dependent on shot noise, which leads to very significant statistical errors.
\subsubsection{Comparison with \texorpdfstring{\ac{PSR}}{PSR} for sums of commuting 2-level generators}
\label{sec:PSR}

Even though the \ac{PSR} in its simplest form is only valid for unitaries with two-level generators, it is straightforward to extend it to generators $H$ which are the sum of commuting two-level generators $H_i$, i.e.
\begin{align}
    H = \sum_{i=1}^{n_\prefac} \prefac_i H_i\,,
\end{align}
where with without loss of generality the eigenvalues of all $H_i$ are $\lambda \in \{0,1\}$ and $n_\prefac$ is the number of generators.

As was shown by Ref.~\cite{Schuld19EvaluatingAnalyticGradients}, \ac{PSR}~(\cref{eq:psr}) together with the product rule of differentiation yields an unbiased estimate of the derivative
\begin{align}
    \delta= \sum_{i=1}^{n_\prefac} \prefac_i\frac{F_{i|+\frac{\pi}{2\prefac_i}}-F_{i|-\frac{\pi}{2\prefac_i}}}{2}\, ,
\end{align}
where $F_{i|+x}$ refers to applying an additional unitary $\e^{\i x\prefac_iH_i}$ during state preparation. 

While these operations may create quantum states outside the ansatz class of the \ac{VQA}, most physical implementations can facilitate them. 
In total $2n_\prefac$ expressions are evaluated. 
If we assume shot noise that is uniform over the parameter space with a variance given by $\sigma_{i\pm}^2=\frac{\sigma^2}{m_{i\pm}}$, where $m_{i\pm}$ is the number of measurements at one position with the total number of measurements $m=\sum_{i=1}^{n_\mu} m_{i+} + m_{i-}$, the optimal choice of measurement distribution is such that $2m_{i+} = 2m_{i-} \coloneqq m_i \propto |\prefac_i|$.
This yields a total error of
\begin{align}
\begin{aligned}
    \edisboth{\epstot^2}=&\sum_{i=1}^{n_\prefac} \frac{\prefac_i^2}{4} \left(2\frac{\sigma^2}{m_{i}/2}\right)\\ \overset{m_i\propto|\prefac_i|}{\longrightarrow}& \frac{\sigma^2\|\vec \prefac\|_{1}^2}{m}\gtrsim \frac{\sigma^2\specwidth^2}{m}\,,
\end{aligned}
\end{align}
where the second step optimizes over the measurement budget $\vec m$ and $\|\vec \prefac\|_{1}$ is an upper bound to the spectral width $\specwidth^2$.
This is the same scaling as for \ac{ULGE}, which is also true for $\Omega$, as this method is also unbiased.

\begin{table*}[ht]
    \def\arraystretch{1.5}
	\begin{tabular}{l||c|c|c|c|c|c}
		&$\edisboth{\epstot^2}$ ($m\rightarrow\infty$) &$\Omega^2$ ($m\rightarrow0 $)& $\sn{\estat^2}/\edisboth{\epstot^2}$ & \#Meas. Pos. $(2n_x)$  & AC? & Priors?  \\\hline\hline
		\ac{BLGE} & $\frac{\sigma^2\specwidth^2}{m}$ &$\frac{m\edistheta{{\delta^2}}}{\sigma^2}\times \frac{1}{\specwidth_{\mathrm{eff}}^2}$& $0 \rightarrow 1$ & $2 \rightarrow 2n_{\mu}$  & Yes &Yes  \\
		\ac{SLGE} & $\propto {m}^{-2/3}$ &$\frac{m\edistheta{{\delta^2}}}{\sigma^2}\times \frac{1}{\nu_{\mathrm{eff}}^2}$& $0 \rightarrow \frac{2}{3}$ & $2$ & Yes &Yes \\
		
		\ac{ULGE} & $\frac{\sigma^2\specwidth^2}{m}$ &$\frac{m\edistheta{{\delta^2}}}{\sigma^2}\times \frac{1}{\specwidth^2}$ &1 & $2n_{\mu}$  & Yes & No \\
		\ac{PSR} & $\frac{\sigma^2\specwidth^2}{m}$ &$\frac{m\edistheta{{\delta^2}}}{\sigma^2}\times \frac{1}{\specwidth^2}$&1 & $2n_{\prefac}$  & No  & No\\
	\end{tabular}
	\caption{Summary of the analyzed methods. 
	The total error refers to the limit of many measurements, while the expression for $\Omega$ is valid in the setting of very few measurements. 
	The notation $a\rightarrow b$ indicates that the value $a$ is valid for small measurement budgets and $b$ is valid for large measurement budgets. 
	AC: Ansatz Class - ``yes'' indicates that all expressions to be evaluated are in the original ansatz class of the \ac{VQA}.
  }
	\label{tab:sum}
\end{table*}

\subsection{\texorpdfstring{\Acl{SLGE}}{Single linear gradient estimator}}
\label{sec:single}
For $m\rightarrow 0$, $\c$ will be chosen small to minimize the statistical error. As the $\ell_1$ norm term heavily penalizes multiple measurements, this leads to a strategy with only a single measurement position method, similar to finite differences. The method where the measurement settings is restricted to $n_x=1$,we call \ac{SLGE}.
The optimization problem~\cref{eq:primalop} for \ac{SLGE} simplifies to
\begin{align}\label{eq:singleop}
	\edisboth{\epstot^2}(\cscalar,x) 
	&= \sum_{k=1}^{n_\mu}\sigakt(\cscalar\sin(\mu_k x)-\mu_k)^2+\frac{\sigma^2}{m}\cscalar^2 \,,
	\end{align}
which describes a quadratic polynomial in $\cscalar$ and a trigonometric polynomial in $x$. which can be minimized numerically exact by finding roots of a polynomial of degree $3\nu$.
We show in \cref{ap:single_err}, that for $m\rightarrow \infty$, \ac{SLGE} scales as
\begin{align}\label{eq:23scaling}
    \edisboth{\epsilon^2}\propto m^{-2/3}\,,
\end{align}
which is outperformed by the $\propto m^{-1}$ scaling of the previous methods. 
This shows that for the best performance for a massive measurement budget, multiple measurement positions are required. 
In contrast for a small measurement budget, we derive the following theorem in \cref{ap:single_om}
\begin{theorem}\label{thm:single_small_m}
For an optimal \ac{SLGE} (and therefore \ac{BLGE}) strategy, in the limit of small $m$, the \om\ ($\Omega$) can be lower-bounded by
\begin{align}
    \Omega^2 &\geq 
    \frac{m\edistheta{{\delta^2}}}{\sigma^2}\times \frac{1}{\specwidth_{\mathrm{eff}}^2} +O(m^2)
\end{align}
with the effective spectral width \begin{align}
    \nu_{\mathrm{eff}}\coloneqq\sqrt{\edis{\Delta^2}/\edis{\delta^2}}\in[1,\nu]\,,
\end{align} the expected ratio between the second derivative $\Delta$ and first derivative $\delta$ of $F(x)$.
\end{theorem}
This theorem shows the strength of the strategy, as $\nu_{\mathrm{eff}}$ also takes into account the relative significance of the individual eigenvalue differences.  Notably, this can be significantly smaller than the spectral width when large eigenvalues are very rare, as is the case in an exponential or Gaussian like distributions with a long but negligible tail.
Thus we expect that in the regime of small $m$, \ac{PSR} and \ac{ULGE} require a factor of $\frac{\nu^2}{\nu^2_{\mathrm{eff}}}$ more measurement rounds for the same quality.
For large $m$, one can ask when \ac{SLGE} will be outperformed by \ac{ULGE} type methods. In the example of used in \cref{fig:measpos}, the crossover occurred only after $\sim 10^5$ measurements and $\Omega\geq 0.996$, which is only visible on the top right plot $\Omega$ after extensive magnification. This means that the regime where \ac{ULGE} takes over is only for very large $m$ where the exact derivative is basically known already. This behavior is not specific to the selected priors, but actually holds for any prior distributions.
\begin{theorem}
\label{thm:single_is_good}
    For any distribution of frequency amplitudes $\mathcal{D}_{\vec \theta}$, single shot noise variance $\sigma^2$ and a measurement budget $m$, there exists an optimal \ac{SLGE} strategy with a \om\ of at least $99\%$ that of an unbiased method. (i.e.\ $\Omega_\mathrm{S}\geq 0.99\,\Omega_{\mathrm{U}}$)
\end{theorem}
The proof is shown in \cref{ap:single_is_good}. We show this by constructing an explicit \ac{SLGE} strategy which achieves this bound for all possible distributions. 
There we also proof the following corollary which shows that even a deterministic \ac{SLGE} method only dependent on the spectral width $\nu$ is already competitive
\begin{corollary}\label{cor:single_is_good}
     An \ac{SLGE} algorithm measuring at position $x=\frac{\pi}{2\nu}$ has a \om\ that is at least $97.5\%$ the \om\ of \ac{ULGE}, regardless of the underlying distribution ($\mathcal{D}_{\vec \theta}$, $\sigma^2$, $m$).
      (i.e.\ $\Omega_\mathrm{S}\geq 0.975\,\Omega_{\mathrm{U}}$)
\end{corollary}
The theorem and corollary show quantitatively that the expected gains from using unbiased estimation methods with multiple measurement positions w.r.t.\ the \om\ ($\Omega$) are small, even for large measurement budgets. It is also worth pointing out that the theorem and corollary do not use the periodicity of the unitary explicitly, meaning they also hold for \emph{non-periodic} unitaries with a generator of spectral width $\nu$.

\subsection{Summary of the measurement budget allocation methods}
\Cref{tab:sum} summarizes the main methods discussed above by comparing their error scalings, the number of expressions that need to be evaluated, whether all expressions are within the ansatz class of the \ac{VQA}, and whether  prior estimates of the shot noise $\sigma^2$ and the second-moment matrix $C_a$ are needed.

\section{Application to \texorpdfstring{\acsp{QAOA}}{QAOAs}}
\label{sec:results}
For our numerics, we use a popular \ac{QAOA} setup \cite{FarGolGut14}.
Here, the ground state of the problem Hamiltonian $H_c$ encodes the solution to the $\MaxCut$ problem on a graph $\mathcal{G} = ([N], E)$ with vertex set $[N]$ and edge set $E$.
The \MaxCut problem is the problem of finding a labelling of the vertex set that maximizes the number of so-called \textit{cut edges}.
The allowed labels are 0 and 1 and an edge is \textit{cut} if it connects two vertices that have different labels.
We identify the computational basis states of our qubits with the two labels.
$H_c$ contains terms for every edge in the graph, which are valued at $-1$ if the edge is cut and $0$ otherwise:
\begin{align}
  \label{eq:hc}
  H_c = \frac{1}{2}\sum_{(i,j) \in E} \left(Z_iZ_j - \mathds{1}\right)\,,\quad E \subseteq [N]\times[N]\,.
\end{align}
The energy of the system described by $H_c$ is minimized by any state that corresponds to the maximal number of edges being cut.
We denote this maximal number of cut edges by $\MaxCut (\mathcal{G})$.

In our numerical experiments the edge set $E$ was randomly generated by selecting a subset of $2N$ edges from the set of all possible edges on $N$ vertices.

The \ac{QAOA} cost function is defined as
\begin{align}
    F(\vec\theta)=\bra{\vec\theta}H_c\ket{\vec\theta}\,,
\end{align}
where the state $\ket{\vec\theta}$ is prepared using the parametrized circuit
\begin{align}
    \ket{\vec\theta} = \left[\prod_{\alpha=1}^{L}U_b(\theta_{2\alpha-1})U_c(\theta_{2\alpha})\right]\ket{+}^{\otimes N}
\end{align}
with $U_c(\theta)=\exp(\i \theta H_c)$ an evolution under $H_c$ and $U_b(\theta)=\exp(\i\theta H_b)$ an evolution under a mixing Hamiltonian $H_b$. 
In our setup the Hamiltonian is given by
\begin{align}
   H_b = -\frac{1}{2}\sum_{i=1}^N X_i\,.
\end{align}
The number of times each unitary appears is referred to as the circuit depth and labelled $L$ and one pair of unitaries $U_b$ and $U_c$ is referred to as one layer in the following.
The initial state that these unitaries are applied to is $\ket{+}^{\otimes N}$, the ground state of $H_b$.
In order to effectively compare the performance of \ac{QAOA} on different graph instances $\mathcal G$ and $\mathcal G'$ with $\MaxCut (\mathcal G) \neq \MaxCut (\mathcal G')$, we introduce the approximation ratio
\begin{align}\label{eq:app_ratio}
    r &= -\frac{F(\vec\theta_{\text{alg}})}{\MaxCut (\mathcal G)} =\frac{F(\vec\theta_{\text{alg}})}{\lambda_{\min}(H_c)}\,,
\end{align}
where $\vec\theta_{\text{alg}}$ is some (possibly intermediate) parameter point determined by some optimization algorithm and $\lambda_{\min} (\argdot)$ evaluates to the smallest eigenvalue of its argument.

\subsection{Deriving prior estimates}
\label{sec:ca}
\begin{figure}[ht]
	\centering
	\includegraphics[width=0.49\linewidth]{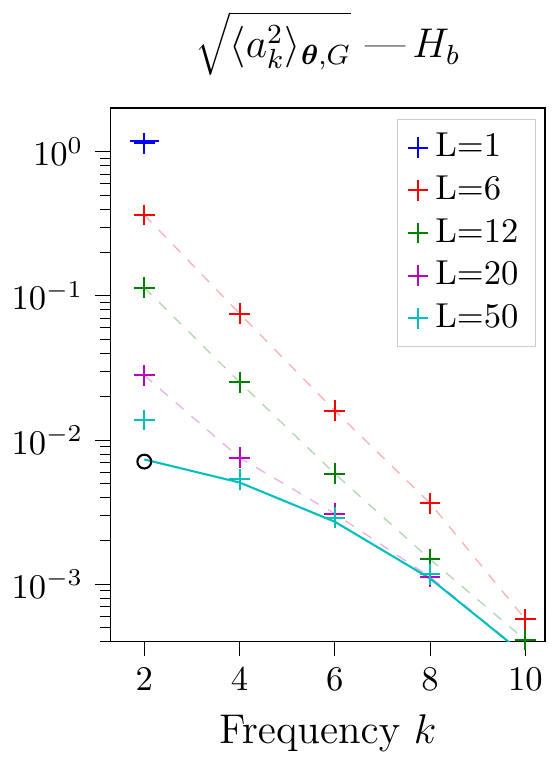}
	\includegraphics[width=0.49\linewidth]{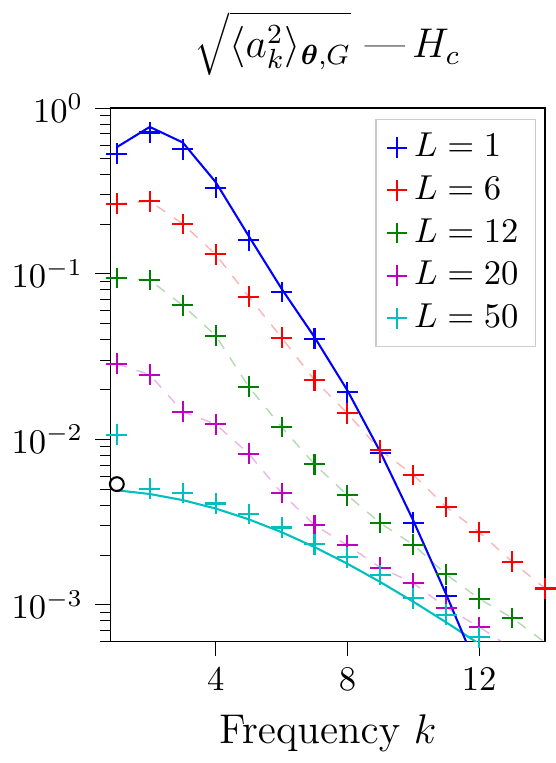}
	\includegraphics[width=0.6\linewidth]{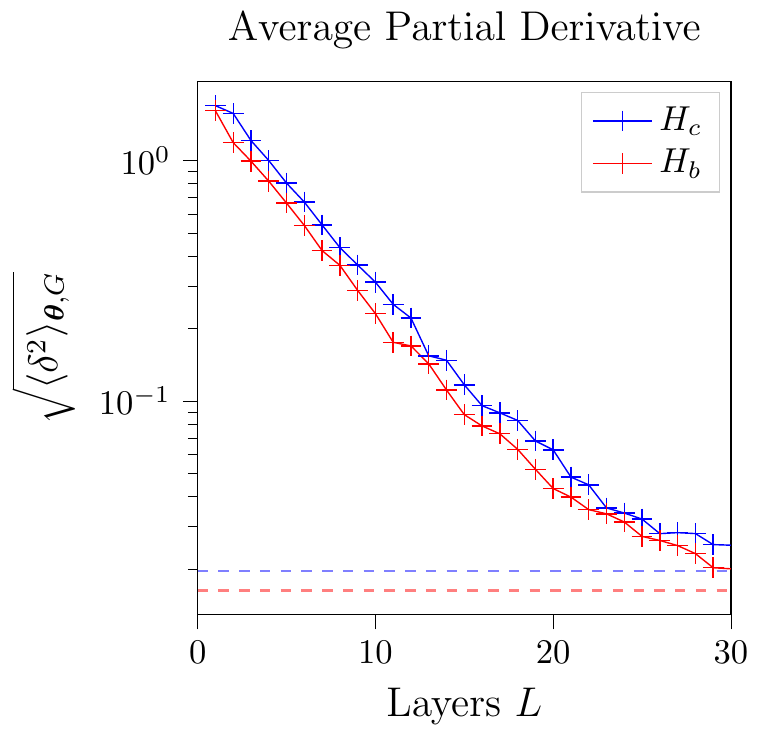}
	\caption{\textbf{Top:} The dependence of $\sigakt$ on the frequency $\mu_k$ for various circuit depths $L$ and a system size $N=16$ with $M=32$ edges as well as a uniform distribution over $\vec \theta$ and all graphs for $2000$ randomly sampled instances.
	The Fourier decomposition is taken w.r.t.\ a layer in the bulk of the circuit, specifically the $\lceil\frac L2\rceil$-th layer. The decomposition is for the $H_b$ gate (left) and the $H_c$ gate (right). The solid lines for $L=1,50$ represent analytical estimates. The black circles show the estimate when we restrict the distribution to graphs which are fully connected and exhibit no non-trivial graph automorphism.
	\textbf{Bottom:} The root of mean square-amplitude of the partial derivative in dependence of the circuit depth. The dashed line shows the Barren plateau limit.
	\label{fig:a_dependence}}
\end{figure}
In this example, $H_c$  is dependent on the graph instance, hence the $\sigakt$ are too.
Finding the exact coefficients for a particular instance can be assumed to be difficult, since already determining the spectral width of $H_c$ is an \NP-hard task in general. It is possible however to sample from graph bipartitions to find an approximate spectral distribution numerically.

To derive analytical estimates, we are extending our distribution to also include all considered graph instances
\begin{align}
    \edisboth{\argdot}\rightarrow \av{\argdot}_{s,\vec\theta,G}\,,
\end{align}
where $G$ indicates drawing the samples from the set of all graphs with $M$ edges and $N$ vertices. 
This generalization allows estimating the priors analytically, which we do in \cref{ap:a_dependence}.
The results are shown in \cref{fig:a_dependence}. 
The plotted points represent a mean over $2,000$ randomly chosen points and graph instances. 
For $L=1$ we derive analytical expressions for both Hamiltonians (dark blue line) in \cref{ap:l1_est}.
While this is a tedious problem, it is efficiently solvable. 
For $U_b$, the specific structure of the \ac{VQA} means that only the Fourier coefficient for $\mu=2$ does not vanish and in general, only even coefficients contribute. 
For $L\rightarrow\infty$ we obtain an analytical estimate (cyan line) using the $2$-design assumption in \cref{ap:linf_est}. 
While this estimate faithfully reproduces the bulk of the frequencies at $L=50$ layers, the empirical estimate for the first frequency is significantly larger than the theoretical prediction.
This discrepancy arises from graph instances which correspond to non-universal gate-set \acp{VQA} instances. 
In particular, this is the case when the graph is not fully connected or has a non-trivial automorphism. 
The additional black circles for the first frequency in the figures shows the empirical estimate obtained when we restrict the set $G$ to only graphs which do not exhibit these properties. 
As can be seen these instances faithfully reproduce the $2$-design prediction.

\Cref{fig:a_dependence}c shows this convergence of the expected  derivative.
Numerically the convergence appears to be close to completed at $L\sim30$ layers, which is consistent with known 2-design convergence results in literature, where depth scales roughly on the order of the number of qubits~\cite{hunter-jones_unitary_2019}.
 The dotted lines show the barren plateau limit when the ergodic argument is used.
 Again due to the existence of graphs with symmetries, the convergence is not exactly to the theoretical 2-design limit.
For the intermediate layers, we can argue (particularly for the case regarding $H_c$), that the overall amplitude of the coefficients decays exponentially while the relative values transition from a behavior like $L=1$ to one closer related to $L\rightarrow\infty$. 
It is our belief that a more rigorous understanding of $t$-design convergence within quantum circuits could help to make more quantitative assessment than we are possible to make at present.
We note that while quantifying the priors may be a worthwhile theoretical pursuit, for practical applications it often suffices to have a rough estimate of the amplitudes, as even with significant over estimation of amplitudes the method will still outperform traditional non-Bayesian schemes. Also since real implementations attempt to avoid the barren plateaus regimes, the gradient along the optimization path may be significantly larger and therefore not fully representative of the behavior of the ensemble average.

\subsection{Numerical results}
In this section we first test the estimation accuracy for the different gradient estimation strategies for randomly selected angles $\vec \theta$.
Afterwards we demonstrate their usefulness for gradient descent based \ac{VQA} optimization. 

\subsubsection{Gradient quality for the different allocation methods} 
\begin{figure}
	\centering
	\includegraphics[width=0.5\linewidth]{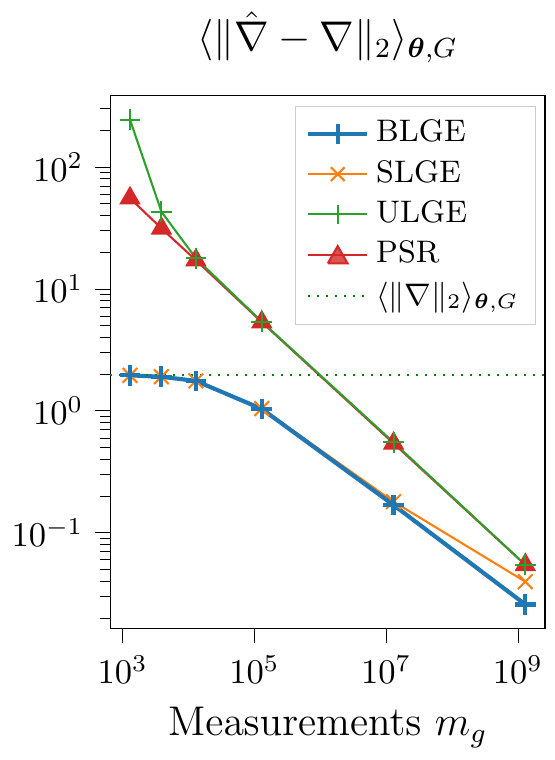}
	\includegraphics[width=0.48\linewidth]{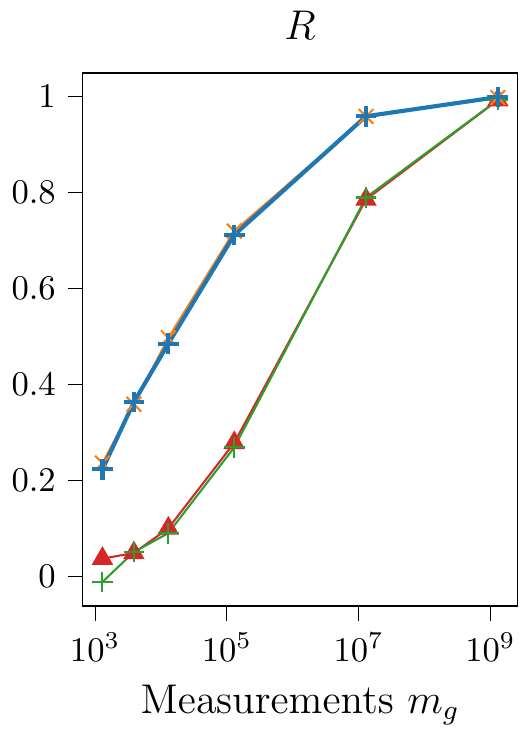}

	\caption{A comparison of the different gradient estimation methods for different measurement budget for the entire gradient ($\mg$).
	The data was generated with a system size of $N=18$, graph instances with $M=36$ edges, and a circuit depth of $L=12$ with $\sigakt$ estimates as stated in \cref{eq:akb_estimate}. 
	Each data point represents an empirical mean over 300 total samples drawn from $30$ different graph instances. 10 random parameter points were evaluated for each instance.
	The measurement budget is distributed in such a way that the $H_c$-layers get twice the measurements of the $H_b$-layers.
	\textbf{Left:} The average 2-norm distance between the exact gradient and the estimate using the different methods.  
	\textbf{Right:} The average relative slope $R$.
	\label{fig:gradqualities}}
\end{figure}
\Cref{fig:gradqualities} shows the quality of the estimated gradients for the different optimization routines and a range of measurement budgets. 
The size of the measurement budget ($m_g$) is the entire budget, i.e.\ jointly for all partial derivatives. The first value of $\mg=1296$ is the first which allows \ac{PSR} and \ac{ULGE} to have exactly one measurement round per expression to be evaluated.
For the priors of the Fourier coefficients, we select 
\begin{align}
\begin{aligned}
\label{eq:akb_estimate}
    \left.\sigakt\right|_b&=10^{-0.3k-1.1}\times\delta_{k\in 2\mathbb{Z}}\\
    \left.\sigakt\right|_c&=10^{-0.3k-1.6}
\end{aligned}
    \end{align}
resulting from a rough exponential fit obtained from \cref{fig:a_dependence}.
\Cref{fig:gradqualities}a depicts the dependence of the 2-norm difference between the exact and the estimated gradient for different budget $m_g$
\begin{align}
    \twonorm{\hat \nabla-\nabla}\,
\end{align}
where $\nabla$ is the exact gradient and $\hat \nabla$ the estimate generated by performing the estimation routine for every component of the gradient. 
We note that for very few measurements \ac{ULGE} performs significantly worse than \ac{PSR}. 
This is because the required positive integer rounding for the individual number of measurements at each site implies a far from optimal allocation for ULGE, while \ac{PSR} does not require any rounding.
Besides this we see good agreement between the numerical results and what was predicted from \cref{fig:measpos}.
Explicitly enforcing the condition that only two position are to be evaluated as in \ac{SLGE} (\cref{sec:single}) only starts to make a significant difference compared to \ac{BLGE} at around $10^7$ measurements which may be already infeasible in a practical experiment. This shows that while the asymptotic behavior is significantly worse, for practical purposes, a correctly chosen finite differences model performs remarkably well.
For fewer measurements, \ac{ULGE} and \ac{PSR} require already $m_g=10^5$ measurements to outperform an estimator which returns the all zero vector. 

For the purpose of gradient descent, one can argue that the direction of the gradient is actually more important than its magnitude.
In order to quantify this notion we investigate the \emph{relative slope} $R$, the ratio between the slope in the direction of the actual gradient and the slope in the direction of the estimated gradient
\begin{align}
    R\coloneqq\av*{\frac{\hat \nabla^T\nabla}{\|\hat \nabla\|_2\|\nabla\|_2}}_{\vec \theta} = \av*{\frac{\mathrm{desc}(\hat \nabla)}{\mathrm{desc}(\nabla)}}_{\vec \theta}
\end{align}
where 
\begin{align}
    \mathrm{desc}(\vec g) = \frac{\mathrm{d}}{\mathrm{d} x} \left[ F\left( \frac{\vec g}{\pnorm[2]{\vec g}} x \right) \right]\,,
\end{align}
which indicates the slope in the direction of $\vec g$.
A value of $R=0$ would indicate that the estimated gradient is orthogonal to the actual gradient. A perfect estimator would achieve a value of $R=1$.
The relative slope $R$ is similar in nature to $\Omega$ defined in~\cref{eq:omega} but differs in the way that the ensemble averages are taken.
Numerically, the manner in which the averages are taken does not qualitatively change our results.
Here again, we see good agreement with the behavior of $R$ and the behavior predicted in \cref{fig:measpos} for $\Omega$. This also justifies why $\Omega$ is indeed a good quality parameter for the estimation.
For the fewest number of measurements considered, \ac{PSR} and \ac{ULGE} struggle to find any decreasing direction, while \ac{BLGE} reliably has a $R=20\%$ overlapp. The other methods require nearly $100$-times the number of measurements for the same quality. Similarly to $\Omega$, \ac{BLGE} and \ac{SLGE} show basically identical performance with regard to the quality measure $R$.
\subsubsection{Parameter Optimization}
We also simulated a complete parameter optimization routine.
Following the proposal from~Zhou et al.~\cite{Zhou2018QuantumApproximate}, we chose the initial parameters $\vec \theta$ to resemble an approximate linear annealing ramp, as this can significantly improve performance compared to random initialization. The exact initialization we used was (even index $H_b$, odd $H_c$)
\begin{align}
    \theta_{i}^{(0)}=\frac{\pi}{20}\left(\frac{4-5{\delta_{i\in 2\mathbb{Z}}}}{L-1} \times i+4\right)
\end{align}
It is worth noting that since we are starting from a specific initialization, the uniformity assumption of our assumed distribution $\mathcal{D}_{\vec \theta}$ might no longer be valid, meaning this also tests the applicability of our approach outside idealized conditions.
For the update step we use a basic gradient descent routine 
\begin{align}
    \vec \theta^{(i+1)}=\vec \theta^{(i)}-\eta^{(i)} \hat \nabla^{(i)}
\end{align}
Since the different methods return gradients with massively different norms, a fixed step size will skew the result heavily. 
To compensate we use a backtracking line-search routine to find a good $\eta^{(i)}$. 
We start with an initial step size that is significantly too large. Then we repeatably measure the observable at the proposed new point $\vec \theta^{(i)}-\eta^{(i)} \hat \nabla^{(i)}$ and either accept the step size if the estimate decreased compared to the current position $\vec \theta^{(i)}$ or half the step size $\eta^{(i)}\rightarrow \eta^{(i)}/2$ and repeat. 
The measurement budget allocated for the line-search estimate is chosen to be the same as for the gradient estimation. 
This line search removes the dependence of the size of the returned gradient without the algorithm becoming a full sweeping algorithm. 

\begin{figure}[ht]
    \centering
    \includegraphics[width=0.473\linewidth]{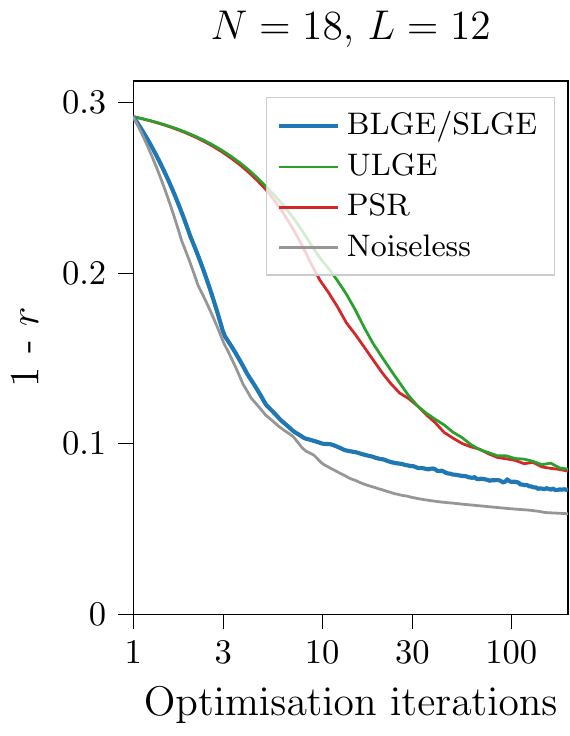}
    \includegraphics[width=0.507\linewidth]{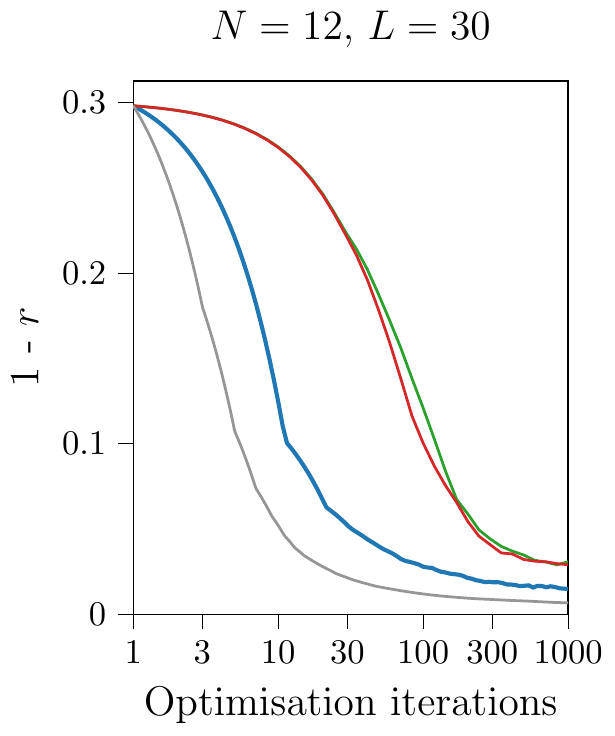}
    \caption{The empirical average approximation ratio $r$ over the iterations of a gradient descent optimization with line search for the different optimization routines. The results are averaged over 23 different problem instances. We note that in the considered measurement regimes \ac{BLGE} and \ac{SLGE} are identical.
    \textbf{Left:} A shallow circuit with $N=18$, $M=36$, $L=12$ with the total measurement for each gradient being $\mg=3888$, 
    which corresponds to $6$ measurements for each generator in the PSR. The priors were taken as an exponential ansatz as given in \cref{eq:akb_estimate}.
    \textbf{Right:} A deeper circuit with $N=12$, $M=24$, $L=30$, $\mg=6480$, 
    which also corresponds to $6$ measurements for each generator. The priors were taken as the barren plateau estimates \cref{eq:a_b_barren}  and \cref{eq:a_c_barren}.
    }
    \label{fig:optimization}
\end{figure}

\Cref{fig:optimization} shows the numerical results of the optimization routine using each of the gradient estimation routines. The cost function shown here is the approximation ratio $r$ as defined in \cref{eq:app_ratio}, which is a common metric for \ac{QAOA} numerics.
We consider both a case with a shallow circuit ($L=12$) with $N=18$ and a deep circuit ($L=30$) in the barren plateau regime, but for a smaller system size of $N=12$. 
In both these cases we consider an overall measurement budget so that each measurement setup for \ac{PSR} is performed $3$ times. This translates to $\mg=3888$ for the shallow circuit case and $\mg=6480$ for the deep circuit.
It is also worth noting that with such a small measurement budget, \ac{BLGE} and \ac{SLGE} are identical. 
Additionally, we plot the case for the exact gradient with an infinite measurement budget in gray. 

\Cref{fig:optimization} also shows that \ac{BLGE} significantly outperforms the unbiased methods with convergence being nearly an order of magnitude faster and also reaching a better minimal value.

\section{Conclusion and outlook}
We have shown that using a Bayesian 
approach in a generalized \ac{PSR} setting can significantly improve the resulting quality of the estimation, which ultimately improves the overall run time and results of the \ac{VQA}. 
In particular, when dealing with barren plateaus, this estimation tool may prove crucial for performance in  practical implementations. 
Our study also makes a strong argument that central difference methods with a reasonably chosen step size is a very good first strategy that is only outperformed by the unbiased \ac{PSR} for large measurement budgets. 

Our work opens up several new research quests. 
\begin{itemize}
    \item In future work, we aim to formally extend our framework to non-periodic unitaries. 

    \item Understanding the second moment matrix and its convergence into the barren plateau regime might allow us to develop strategies to mitigate its effects allowing \acp{VQA} to be effective for more ansatz classes. 
    Such improved priors might come from studies of approximate unitary 2-designs. 
    
    \item 
    Improved optimization methods, such as natural gradient estimation 
    \cite{Stokes2020QuantumNaturalGradient,Yamamoto19OnTheNatural,Wierichs2020AvoidingLocalMinima}
    or higher order optimization procedures requiring second order derivatives, 
    may also benefit from introducing prior assumptions into their estimation routines.

    \item In particular, for practical applications where quantum computation is expensive and classical computation is cheap, using an adaptive estimator may be advantageous. 
    That is, updating the Fourier priors along the optimization path may also help to improve the gradient quality further. 
\end{itemize}

\section{Aknowledgement}
We thank 
Raphael Brieger, 
Lucas Tendick, 
Markus Heinrich, 
Juan Henning,
Michael J.\ Hartmann, Nathan McMahon, and Samuel Wilkinson,
for fruitful discussions. 
This work has been funded by the German Federal Ministry of Education and Research (BMBF) within
the funding program ``quantum technologies -- from basic research to market'' via the joint project MANIQU
(grant number 13N15578) and the Deutsche Forschungsgemeinschaft (DFG, German Research Foundation) under the grant number 441423094 within the Emmy Noether Program. 


\section{Acronyms and list of symbols}
\label{sec:symbols}

\begin{acronym}[SBLGE]\itemsep.5\baselineskip
\acro{AGF}{average gate fidelity}

\acro{BLGE}{Bayesian linear gradient estimator}

\acro{BOG}{binned outcome generation}

\acro{RB}{randomized benchmarking}

\acro{CP}{completely positive}
\acro{CPT}{completely positive and trace preserving}

\acro{DFE}{direct fidelity estimation} 

\acro{GST}{gate set tomography}

\acro{HOG}{heavy outcome generation}

\acro{LGE}{linear gradient estimator}

\acro{MUBs}{mutually unbiased bases} 
\acro{MW}{micro wave}

\acro{NISQ}{noisy and intermediate scale quantum}

\acro{POVM}{positive operator valued measure}
\acro{PQC}{parametrized quantum circuit}
\acro{PSR}{parameter shift rule}
\acro{PVM}{projector-valued measure}

\acro{QAOA}{quantum approximate optimization algorithm}

\acro{SFE}{shadow fidelity estimation}
\acro{SIC}{symmetric, informationally complete}
\acro{SLGE}{single Bayesian linear gradient estimator}
\acro{SPAM}{state preparation and measurement}

\acro{QPT}{quantum process tomography}

\acro{TT}{tensor train}
\acro{TV}{total variation}

\acro{ULGE}{unbiased linear gradient estimator}

\acro{VQA}{variational quantum algorithm}
\acro{VQE}{variational quantum eigensolver}

\acro{XEB}{cross-entropy benchmarking}

\end{acronym}


\begin{itemize}[labelindent=3em,leftmargin=*]
\item[$N$:] number of qubits/vertices (\acs{QAOA}) 
\item[$L$:] number of layers of the \ac{VQA}
\item[$\vec \lambda$:] eigenvalues of the gate generator
\item[$\vec \mu$:] eigenvalue differences of the gate generator
\item[$a_k,b_k,c_k$:] Fourier coefficients
\item[$C_a$:]  Second moment matrix of the Fourier coefficient $a_k$
\item[$\sigma$:] single shot shot noise
\item[$w_i$:] weights of linear estimator
\item[$\vec\theta$:] \acs{VQA} parameters 
\item[$\theta = x_i$:] measurement position for given $\theta\in \vec\theta$
\item[$n_x$:] number of (positive) measurement positions
\item[$m_i$:] number of measurement rounds for $x_i$
\item[$m$:] total number $m=\sum m_i$
\item[$\nu$:] largest difference between two eigenvalues of the unitary generator, i.e.\ $\max_k\{\mu_k\}$
\item[$\delta$:] derivative of the cost function w.r.t.\ one parameter $\theta$
\item[$\Delta$:] second derivative of the cost function w.r.t.\ one parameter $\theta$
\end{itemize}

\bibliographystyle{myapsrev4-2}
\bibliography{mk,Len.bib}
\newpage
\onecolumn
\section*{Appendices}
\appendix
\setcounter{section}{1}
\renewcommand\thesubsection{\Alph{subsection}}

\subsection{Solving the Bayesian allocation problem} 
Here we derive the theoretic underpinnings of the Bayesian allocation method. 
Namely, we derive the effective dual that is used for the numerical implementation and proof \cref{thm:finite_positions}, which concerns the sparsity of the solution. 

\subsubsection{Finding the dual problem}
\newcommand{\slackvar}{z}
\label{sec:optimization}
In this section, we derive the dual formulation, we use to find the optimal measurement positions as explained in the main text.
The minimization (\cref{eq:primalop}), we are interested in, can be rewritten into a constraint problem
\begin{align}
    \min_{\vec x \in \RR^{n_x},\c\in\RR^{n_x}}\edisboth{\epstot^2}(\vec x, \c)&=\min_{\vec x\in \RR^{n_x},\c \in \RR^{n_x}} \sum_{k=1}^{n_\mu} \sigakt\left(\sum_{i=1}^{n_x} \ci\sin(x_i \mu_k)-\mu_k\right)^2+ \frac{\sigma^2}{m} \norm{\c}_1^2\\
    &=\min_{\vec x\in \RR^{n_x},\c \in \RR^{n_x},\vec\slackvar\in \RR^{n_\mu},l\in\RR} \sum_{k=1}^{n_\mu} \sigakt \slackvar_k^2+ \frac{\sigma^2}{m} l^2\\
    &\mathrm{s.t.} \quad  \sum_{i=1}^{n_x} \ci\sin(x_i \mu_k)-\mu_k=\slackvar_k \, ,\, \norm{\c}_1\leq l \, ,
\end{align}
where we introduced the variables $\vec \slackvar$ and $l$. If we keep $\vec x$ fixed (it can be assumed to describe a fine grid covering a complete period of the function), this describes a convex optimization problem.
We proceed to derive the dual $g$.
\begin{align}
    g(\vec x;\vec \slackvar,\c,l;\vec \slvar,\tau)& =\sum_{k=1}^{n_\mu} \sigakt \slackvar_k^2+ \frac{\sigma^2}{m} l^2+\sum_{k=1}^{n_\mu}2\slvar_k \left( \sum_{i=1}^{n_x} \ci\sin(x_i \mu_k)-\mu_k-\slackvar_k\right)-2\tau(l-\norm{\c}_1)\\
    \frac{\partial g}{\partial \slackvar_k}&=2 \sigakt \slackvar_k-2\slvar_k \rightarrow \slackvar_k=\frac{\slvar_k}{\sigakt}\\
    \frac{\partial g}{\partial l}&=2 \frac{\sigma^2}{m} l-2\tau \rightarrow l=\tau\frac{m}{\sigma^2}\\
    \frac{\partial g}{\partial \ci}&= 2\tau \,\mathrm{sgn}(\ci) +2\sum_{k=1}^{n_\mu}\slvar_k \sin(x_i \mu_k)\rightarrow 
    \tau\geq \left|\sum_{k=1}^{n_\mu}\slvar_k \sin(x_i \mu_k)\right|\label{eq:prim_const}
\end{align}
where the last step follows as $g$ is affine w.r.t.\ $\ci$ for the positive and negative axis. \Cref{eq:prim_const} also implies that either $\ci=0$ or $\tau=|\sum_{k=1}^{n_\mu}\slvar_k \sin(x_i \mu_k)|$, also known as complementary slackness.
\begin{align}
    g(\vec x;\vec \slvar,\tau)&=-\sum_{k=1}^{n_\mu} \frac{\slvar_k^2}{\sigakt}-\frac{m}{\sigma^2} \tau^2-2\sum_{k=1}^{n_\mu}\slvar_k \mu_k\\
    &\mathrm{s.t.}\quad \forall x_i :\quad\tau\geq |\sum_{k=1}^{n_\mu}\slvar_k \sin(x_i \mu_k)|
\end{align}
Since $\vec x$ wants to minimize this expression, $\tau$ needs to be maximized. As all measurement positions are allowed and we impose no bound on $n_x$ meaning $\vec x$ can be distributed arbitrarily dense, it follows
\begin{align}
    \max_{i\in[n_x]} |\sum_{k=1}^{n_\mu}\slvar_k \sin(x_i \mu_k)|\rightarrow\norm{\sum_{k=1}^{n_\mu} \slvar_k \sin(\mu_k (\argdot))}_{\infty}\,
\end{align}
where the infinity norm refers the absolute value maximum over the domain of the function,
and we get $\tau =\norm{\sum_{k=1}^{n_\mu} \slvar_k \sin(\mu_k (\argdot))}_{\infty}$, which leads to 
\begin{align}
    g(\vec \slvar)&=-\sum_{k=1}^{n_\mu} \frac{\slvar_k^2}{\sigakt}-\frac{m}{\sigma^2} \norm{\sum_{k=1}^{n_\mu} \slvar_k \sin(\mu_k (\argdot))}_{\infty}^2-2\sum_{k=1}^{n_\mu}\slvar_k \mu_k\,.
\end{align}
We conclude that 
\begin{align}
    \edisboth{\epstot^2}^*&=\max_{\vec \slvar\in \RR^{n_\mu}} g(\vec \slvar)\,,
\end{align}
since the minimization of the primal is convex for fixed $\vec x$.
Solving this requires maximizing a concave problem.
 $g$ has the same solution as minimizing
\begin{align}
    \tilde g(\vec \slvar)&=\sum_{k=1}^{n_\mu}\frac{\sigma^2}{m\sigakt} \slvar_k^2+ \norm{\sum_{k=1}^{n_\mu} \slvar_k \sin(\mu_k (\argdot))}_{\infty}^2-2\sum_{k=1}^{n_\mu}\slvar_k \mu_k
\end{align}
with $\tilde g=-\frac{m}{\sigma^2}g$ and $\slvar^*\rightarrow -\frac{\sigma^2}{m}\slvar^*$ which we find has better numerical stability, especially for large measurement budgets $m$.
Via complementary slackness, the final measurement positions $\vec x^*$ are just the global maxima positions of $\rho_{\vec \slvar}(x)=|\sum_{k=1}^{n_\mu}\slvar_k^* \sin(x \mu_k)|$, which can be obtained by solving a trigonometric polynomial, which has a most $\nu$ maxima. To obtain $\c^*$, we solve the original problem for the now fixed measurement positions $x^*$.

\subsubsection{Proof of \texorpdfstring{\cref{thm:finite_positions}}{theorem 1}}
\label{ap:proof_of_sparse}
Here we proof that the optimization problem has a sparse solution. Explicitly that $n_x=n_{\mu}$ positions suffices.
For this we assume that we have set of measurement positions $\vec x \in [0,\pi)^{n_x}$. For the problem we assume to have
\begin{align}
    \edisboth{\epstot^2}&=(\S\c-\vec \mu)^T C_a(\S\c-\vec \mu)+\frac{\sigma^2}{m}\onenorm*{\c}^2\,,
\end{align}
and an optimal solution $\vec w^*$. As such we can define $\edisboth{\esys^2}^*\coloneqq(\S\c^*-\vec \mu)^T C_a(\S\c^*-\vec \mu)$ and $y^*=\S\c^*$
This means the problem is equivalent as
\begin{align}
    \edisboth{\epstot^2}&=\edisboth{\esys^2}^*+\frac{\sigma^2}{m}\left(\min_{\c:y^*=\S\c} \onenorm*{\c}\right)^2\,,
\end{align}
The latter term constrains an $\ell_1$-norm optimization with $n_{\mu}$ linear constraints. It is well known that there exists an optimal sparse solution where the number of non-vanishing entries is at most the number of constraints, which are known as  \emph{basic feasible solutions} in LP literature. As this holds regardless of the actual value $\vec y^*$, there also exists a sparse solution for the entire problem which proofs the theorem.
\qed 

The same also holds for the unbiased case, where $\vec y^*=\vec\mu$ a strict requirement.

\subsection{%
\texorpdfstring{\Acf{SLGE}}{SLGE}
}
In the following section we derive various properties of our single measurement estimation strategy.
In \cref{ap:pre}, we derive the optimal coefficient $\cscalar$ for our estimator, which we use to rewrite the expected total error of our estimator as a function of only the measurement position $x$.
Next in \cref{ap:single_err}, we proof the scaling in the limits of many measurements from \cref{eq:23scaling}. In \cref{ap:single_om}, we proof \cref{thm:single_small_m} which is valid for very few measurements.
Finally, in \cref{ap:single_is_good} we prove \cref{thm:single_is_good} and \cref{cor:single_is_good} from our main text which are about the relative performance of \ac{SLGE} and \ac{ULGE}.

\subsubsection{Preliminaries} 
\label{ap:pre}
Recall that we can express the expected total error as
\begin{align}
	\edisboth{\epstot^2}(\cscalar,x) 
	&= \sum_{k=1}^{n_\mu}\sigakt(\cscalar\sin(\mu_k x)-\mu_k)^2+\frac{\sigma^2}{m}\cscalar^2\\
	&\eqqcolon \expvar{(\cscalar\sin(\mu x)-\mu)^2} +\frac{\sigma^2}{m}\cscalar^2\label{eq:toterr_single}\,,
\end{align}
where we introduced the shorthand for the ensemble average
\begin{align}
    \expvar{g(\mu)}\coloneqq \sum_k^{n_{\mu}} \sigakt g(\mu_k)\,,
\end{align}
for an arbitrary function $g$. 

Given $x$, the optimal coefficient $\cscalar$ can be determined to be 
\begin{align}
    \label{eq:opt_w_single}
    \cscalar^* = \frac{\expvar{\mu\sin(\mu x)}}{\expvar{\sin^2(\mu x)}+\frac{\sigma^2}{m}}\,.
\end{align}
Plugging $\cscalar^*$ into~\cref{eq:toterr_single} yields
\begin{align}
\label{eq:epstotsingle-solveforx}
	\edisboth{\epstot^2}(x) &= \expvar{\mu^2} - \frac{\expvar{\mu \sin(\mu x)}^2}{\expvar{\sin^2(\mu x)}+\frac{\sigma^2}{m}}
	=\frac{\expvar{\mu^2}\expvar{\sin^2(\mu x)}-\expvar{\mu \sin(\mu x)}^2+\expvar{\mu^2}\frac{\sigma^2}{m}}{\expvar{\sin^2(\mu x)}+\frac{\sigma^2}{m}}
\end{align}
\subsubsection{Numerical optimization}
\label{ap:numsing}
We are now going to outline how to numerically choose the optimal value of $x$.
For the numerical minimization we use \cref{eq:epstotsingle-solveforx}. 
Here we assume that estimates of $\sigakt$ and $\sigma$ are known.

The derivative of \cref{eq:epstotsingle-solveforx} w.r.t.\ $x$ is
\begin{align}
    \partial_x\edisboth{\epstot^2}(x) &= - \frac{2\expvar{\mu \sin(\mu x)}\expvar{\mu^2 \cos(\mu x)}\left(\expvar{\sin^2(\mu x)}+\frac{\sigma^2}{m}\right)-\expvar{\mu \sin(\mu x)}^2\expvar{2\mu\cos(\mu x)\sin(\mu x)}}{\left(\expvar{\sin^2(\mu x)}+\frac{\sigma^2}{m}\right)^2}\\
    &=-\frac{2\expvar{\mu \sin(\mu x)}}{\left(\expvar{\sin^2(\mu x)}+\frac{\sigma^2}{m}\right)^2}\times\underbrace{\left(\expvar{\mu^2 \cos(\mu x)}\left(\expvar{\sin^2(\mu x)}+\frac{\sigma^2}{m}\right)-\expvar{\mu \sin(\mu x)}\expvar{\mu\cos(\mu x)\sin(\mu x)}\right)}_{\eqqcolon h(x)}\nonumber\,.
\end{align}
The relevant minima candidates are therefore the roots of the second factor - labelled $h(x)$ - since roots of the first factor always yield a maximum of $\edisboth{\epstot^2}(x)$ (as the first factor appears in \cref{eq:epstotsingle-solveforx} with a negative sign).

We can rewrite $h(x)$ explicitly as
\begin{align}
    h(x)\label{eq:x_fit}
    &=\expvar{\mu^2 \cos(\mu x)}\left(\expvar{\sin^2(\mu x)}+\frac{\sigma^2}{m}\right)-\frac{1}{2}\expvar{\mu \sin(\mu x)}\expvar{\mu\sin(2\mu x)}\\
    &=\sum_{k,l}\sigakt\sigalt\mu_k^2\cos(\mu_k x)\sin^2(\mu_l x)
-\frac12\sum_{k,l} \sigakt \sigalt \mu_k\mu_l \sin(\mu_k x)\sin(2\mu_l x)
	+\frac{\sigma^2}{m}\sum_k \sigakt \mu_k^2\cos(\mu_k x)\,,\nonumber
\end{align}
which can be solved efficiently using standard solvers as it only requires finding the roots of a trigonometric polynomial of degree $3\specwidth$. From these candidate solutions, we can find the numerical exact optimal solution.
\subsubsection{Limit for many measurements}
\label{ap:single_err}
In this section we determine the scaling of the expected total error in the limit of many measurements, i.e.\ high measurement accuracy, as stated in \cref{tab:sum}.

For $m\rightarrow\infty$, the optimal measurement position becomes $x\rightarrow 0$, since noiseless measurements lead to a finite difference approximation which becomes exact when the measurement position approaches $0$.
This observation justifies the use of a Taylor series expansion
\begin{align}
\expvar{\mu^2}\expvar{\sin^2(\mu x)}-\expvar{\mu \sin(\mu x)}^2=\frac{\xi}{2} x^6 +O(x^8)\,,
\end{align}
where $\xi=\frac{\expvar{\mu^2}\expvar{\mu^6}-\expvar{\mu^4}\expvar{\mu^4}}{18}$, which is non-negative. 
Plugging this expression into \cref{eq:epstotsingle-solveforx} yields

\begin{align}
    \edisboth{\epstot^2}(x) &=\frac{\frac{\xi}{2} x^6+\expvar{\mu^2}\frac{\sigma^2}{m} + O(x^8)}{\expvar{\mu^2}x^2 + \frac{\sigma^2}{m}+O(x^4)}\\
    &=\frac{\frac{\xi}{2} x^6 + \expvar{\mu^2} \frac{\sigma^2}{m}}{\expvar{\mu^2}x^2}+O(m^{-1},x^6)\,
\end{align}
The expression is minimized by 
\begin{align}
    x^* &= \frac{\expvar{\mu^2}^{1/6}}{\xi^{1/6}}\left(\frac{\sigma^2}{m}\right)^{1/6}\,,\\
    \edisboth{\epstot^2}(x^*)
    &\rightarrow\frac{3\xi^{1/3}}{2 \expvar{\mu^2}^{1/3}}\left(\frac{\sigma^2}{m}\right)^{2/3}\propto m^{-2/3}\,.
   \end{align}
For $m\rightarrow\infty$, taking the limit of \cref{eq:opt_w_single}, the single measurement scheme will converge to a 
a standard central differences method  with coefficient $\cscalar^*\sim \frac{1}{x^*}$.
Therefore, the expected statistical error $\edisboth{\estat^2} = \cscalar^2\frac{\sigma^2}{m}$ becomes 
\begin{align}
    \edisboth{\estat^2}\approx\frac{\sigma^2}{mx^{*2}}=\frac{2}{3} \edisboth{\epstot^2}\,,
\end{align}
meaning that for $m\rightarrow\infty$ the statistical error makes up $2/3$ of the total error.

\subsubsection{%
Limit for few measurements \texorpdfstring{-- proof of \cref{thm:single_small_m}}{}
}

\label{ap:single_om}
We are now going to turn our attention toward the regime of few measurements.
In this scenario statements about the scaling of the total expected error do not make much sense since the \ac{SLGE} strategy tends to return very small gradient estimates in the case of low measurement accuracy.
This means that the total expected error becomes the expected magnitude of the partial derivative 
($\edisboth{\epstot^2}=\edistheta{\delta^2}$), 
where it plateaus, as can be seen in \cref{fig:gradqualities}.
Instead, the \om\ describes a useful quantity in this regime.
One can straightforwardly derive these simple expressions.
\begin{align}
    \label{eq:deriv_brackets}
    \edis{\delta^2 }&=\expvar{\mu^2}\,,\\
    \label{eq:cov_brackets}
    \edisboth{\hat \delta \delta} &=\cscalar\expvar{\mu\sin(\mu x)}\,,\\
    \label{eq:estim_brackets}
    \edisboth{\hat{\delta}^2} &=\cscalar^2\expvar{\sin^2(\mu x)}+\frac{\sigma^2}{m}\cscalar^2\,.
\end{align}
Also, for the second derivative of the cost function $F$ which we denote as $\Delta$ we get
\begin{align}\label{eq:deriv2_brackets}
    \edis{\Delta^2}\coloneqq\edis{(\partial_x^2 F(0))^2}=\expvar{\mu^4}\,,
\end{align}
which uses that $\sigakt=\edis{b_k^2}$ resulting from the shift invariance assumption of $\mathcal{D}_{\vec \theta}$.

To proof \cref{thm:single_small_m}, proceed to derive a lower bound to the \om\ ($\Omega$) as defined in \cref{eq:omega}.
For $m\ll\sigma^2/\expvar{\sin^2(\mu x)}$, we can approximate $\Omega^2$ as 
\begin{align}
    \Omega^2 &= \frac{\expvar{\mu \sin(\mu x^*)}^2}{\expvar{\mu^2}\left(\expvar{\sin^2(\mu x^*)}+\frac{\sigma^2}{m}\right)}\\
    &=m\frac{\expvar{\mu \sin(\mu x^*)}^2}{\expvar{\mu^2}\sigma^2} +O(m^2)\,.
\end{align}
using the following \cref{lem:mu_av} for the last step.
\begin{align}
\label{eq:omega_lowerbound}
    \Omega^2 &\geq m\frac{\expvar{\mu^2}^2}{\sigma^2\expvar{\mu^4}} +O(m^2)\\
    &= m\frac{\edis{\delta^2}^2}{\sigma^2\edis{\Delta^2}} +O(m^2)\,,
\end{align}
Here we also inserted the definitions~\cref{eq:deriv_brackets} and~\cref{eq:deriv2_brackets}
\qed
\begin{lemma}\label{lem:mu_av}
For any ensemble average over $\mu\geq 0$ with $\expvar{\mu}\neq0$, the following statement holds
\begin{align}
\label{eq:to_show}
    \max_{x\geq 0}\expvar{\mu \sin(\mu x)}\geq  \sqrt{\frac{\expvar{\mu^2}^3}{\expvar{\mu^4}}}\,.
\end{align}
\end{lemma}

\begin{proof}
In order to prove that the maximum satisfies this inequality, it suffices to show that the inequality is satisfied at some point $x^*$.
We choose to show this for $x^*=\frac{\pi}{2}\sqrt{\frac{\expvar{\mu^2}}{\expvar{\mu^4}}}$. 
To this end, we demonstrate that for $x\in [0, x^*]$
\begin{align}
\label{eq:to_show_more_complex}
    \expvar{\mu \sin\left(\mu x\right)}\geq \sqrt{\frac{\expvar{\mu^2}^3}{\expvar{\mu^4}}}\sin\left(\sqrt{\frac{\expvar{\mu^4}}{\expvar{\mu^2}}}x\right)
\end{align}

holds, which gives the desired result for $x=x^*$.
As the inequality is satisfied for $x=0$, we will show the general case by verifying that the left-hand side has a larger derivative than the right-hand side everywhere in this interval which implies the original claim. 

Therefore, it remains to show that
\begin{align}
    \frac{\expvar{\mu^2 \cos\left(\mu x\right)}}{\expvar{\mu^2}} \geq \cos\left(\sqrt{\frac{\expvar{\mu^4}}{\expvar{\mu^2}}}x\right)\geq 0\,.
\end{align}
By defining $p_k\coloneqq\frac{\sigakt\mu_k^2}{\sum_k\sigakt\mu_k^2}=\frac{\sigakt\mu_k^2}{\expvar{\mu^2}}$ this simplifies
to 
\begin{align}
    \sum_i p_i \cos\left(\mu_i x\right)\geq \cos\left(\sqrt{\sum p_k \mu_k^2} x\right)\geq 0
\end{align}
for $x\in\left(0, \frac{\pi}{2\sqrt{\sum_k p_k \mu_k^2}}\right]$.
Further substituting $x'\coloneqq x \sqrt{\sum p_k \mu_k^2}$ as well as $\mu_i'\coloneqq \frac{\mu_i}{\sqrt{\sum p_k \mu_k^2}}$ yields
\begin{align}
\label{eq:derivative_ineq_simplified}
    \sum_i p_i \cos\left(\mu_i' x'\right)\geq \cos\left(x'\right)\geq 0
\end{align}
for $x'\in[0,\frac{\pi}{2}]$ with $\sum p_k \mu_k'^2=1$.
This inequality needs to hold for all $\vec p>\vec 0$, $\vec {\mu'}> \vec 0$ with $\sum_k p_k=1$. 
In the interest of readability, we are going to drop the primes again from here on out.

To prove this final inequality, we reformulate the left-hand side as a minimization problem
\begin{align}
    F(\vec {\mu})=\sum_i p_i \cos\left(\mu_i x\right)+\frac{x^2}{2}\lambda \left(\sum_i p_i\mu_i^2 -1\right)
\end{align}
with Lagrange multiplier $\frac{x^2}{2}\lambda$ for the constraint $\sum_i p_i\mu_i^2=1$. 
This gives the partial derivatives
\begin{align}
    \frac{\partial F}{\partial \mu_i}&=-p_i x \sin\left(\mu_i x\right)+\lambda x^2 p_i\mu_i\,,\\
    \frac{\partial ^2F}{\partial \mu_i^2}&=-p_i x^2 \cos\left(\mu_i x\right)+\lambda x^2 p_i\,,
\end{align}
meaning that an extremal point satisfies
\begin{align}
    \lambda&=\frac{x\sum_ip_i \sin(\mu_ix)}{x\sum_ip_i\mu_ix}<1\,.
\end{align}
as well as either
\begin{align}
    \mu_i&=0 \text{, or }\\ \mathrm{sinc}(\mu_ix)&=\lambda\,.
\end{align}

Since $\left.\frac{\partial ^2F}{\partial \mu_i^2}\right|_{\mu_i=0}=p_ix^2(\lambda-1)$ is negative, $\mu_i=0$ does not describe a local minimum. 
Therefore, $\mu_i>0$ holds. 
Since $\sum_i p_i\mu_i^2=1$ and $\sum_i p_i = 1$, there exists a $\mu_\alpha\leq1$, meaning $\lambda=\mathrm{sinc}(\mu_\alpha x)\geq \mathrm{sinc}(\frac{\pi}{2})=\frac{2}{\pi}$, which implies that all $\mu_i$ have the same value since $\mathrm{sinc}$ is injective on the considered interval. 
With the constraint, this yields $\mu_i\equiv 1$. 

The second derivative at this point is 
\begin{align}
    \left.\frac{\partial ^2F}{\partial \mu_i^2}\right|_{\vec \mu=\vec 1}=p_ix^2 (\lambda -\cos(x))=p_i(x\sin(x)-x^2 \cos(x))> 0
\end{align}
for $0< x\leq \frac{\pi}{2}$, which - as the Hessian is diagonal - means that $\vec\mu=\vec 1$ is indeed the global minimum of $F(\vec \mu)$ in the considered interval.

Since this minimal $\vec \mu$ satisfies \cref{eq:derivative_ineq_simplified}, this concludes the proof of the lemma.
\end{proof}

\subsubsection{Proof of 
\texorpdfstring{\cref{thm:single_is_good} and \cref{cor:single_is_good}}{theorem 4 and corollary 1}
}
\label{ap:single_is_good}
We derive a bound on the optimal $\frac{\Omega^2_{\mathrm{S}}}{\Omega^2_{\mathrm{UB}}}$. This expression can be written as
\begin{align}
    \frac{\Omega^2_{\mathrm{S}}}{\Omega^2_{\mathrm{UB}}} &= \frac{\expvar{\mu \sin(\mu x)}^2}{\expvar{\mu^2}\left(\expvar{\sin^2(\mu x)}+\frac{\sigma^2}{m}\right)}\times \frac{\expvar{\mu^2}+\frac{\nu^2\sigma^2}{m}}{\expvar{\mu^2}}\\
    &=\frac{\frac{\expvar{\mu \sin(\mu x)}^2}{\expvar{\mu^2}^2}\left(1+\frac{\nu^2\sigma^2}{m\expvar{\mu^2}}\right)}{\frac{\expvar{\sin^2(\mu x)}}{\expvar{\mu^2}}+\frac{\sigma^2}{m\expvar{\mu^2}}}\, .
\end{align}
By defining a probability vector $p_k=\sigakt \frac{\mu_k^2}{\expvar{\mu^2}}$ and $\alpha=\frac{\nu^2 \sigma^2}{m\expvar{\mu^2}}$
we get 
\begin{align}
    \frac{\Omega^2_{\mathrm{S}}}{\Omega^2_{\mathrm{UB}}}&=\frac{\left(\sum_k p_k\frac{\sin(x\mu_k)}{\mu_k}\right)^2(1+\alpha)}{\sum_k p_k\left(\frac{\sin(x\mu_k)}{\mu_k}\right)^2+\frac{\alpha}{\nu^2}}\\
    &=\frac{\left(\sum_k p_k\frac{\nu\sin(x\mu_k)}{\mu_k}\right)^2(1+\alpha)}{\sum_k p_k\left(\frac{\nu\sin(x\mu_k)}{\mu_k}\right)^2+\alpha}
\end{align}
by substituting $\mu_k\rightarrow \mu_k \nu$ and $x\rightarrow x\nu$, we have the requirement that $\alpha\geq0$ and $\mu_k\in [0,1]$.
and
\begin{align}
    \frac{\Omega^2_{\mathrm{S}}}{\Omega^2_{\mathrm{UB}}}
    &=\frac{\left(\sum_k p_k\frac{\sin(x\mu_k)}{\mu_k}\right)^2(1+\alpha)}{\sum_k p_k\left(\frac{\sin(x\mu_k)}{\mu_k}\right)^2+\alpha}\\
   \end{align}
or written as a problem we want to find the distribution for the worst \om\ ratio using the best \ac{SLGE} strategy.
\begin{align}
     \left(\frac{\Omega^2_{\mathrm{S}}}{\Omega^2_{\mathrm{UB}}}\right)^*&=\min_{\substack{\alpha\geq 0,\,\vec \mu\in [0,1]^{n_\mu}\\ \vec p \in [0,1]^{n_\mu},\,\sum_i p_i=1}} \left(\max_{x(\alpha,\vec \mu,\vec p)}\left(\frac{\Omega^2_{\mathrm{S}}}{\Omega^2_{\mathrm{UB}}}\right)\right)\\
     &\geq
     \min_{\substack{\alpha\geq 0}} \left(\max_{x(\alpha)}\left(\min_{\substack{\,\vec \mu\in [0,1]^{n_\mu}\\ \vec p \in [0,1]^{n_\mu},\,\sum_i p_i=1}}\left(\frac{\Omega^2_{\mathrm{S}}}{\Omega^2_{\mathrm{UB}}}\right)\right)\right)\, ,
     \end{align}
where we changed to order of optimization to find a lower bound. To solve the innermost bracket, we observe that by defining $\tau(\mu)=\frac{\sin(x\mu)}{\mu}$ as a random variable that $\sum_k p_k\frac{\sin(x\mu_k)}{\mu_k}=\av{\tau}_{\vec p}$ and $\sum_k p_k\left(\frac{\sin(x\mu_k)}{\mu_k}\right)^2=\av{\tau^2}_{\vec p}$ are the first and second moment of $\tau$ with $\av{}_{\vec p}$ the expectation value of the distribution spanned by $\vec p$. This means for the expression
\begin{align}
    \frac{\Omega^2_{\mathrm{S}}}{\Omega^2_{\mathrm{UB}}}
    &=\frac{\av{\tau}_{\vec p}^2(1+\alpha)}{\av{\tau^2}_{\vec p}+\alpha}\\
   \end{align}
We also note that if we restrict $x\in [0,\pi/2]$, $\tau\in [\sin(x),x]$, meaning $\tau$ is a bounded variable. We note that for a given expectation value $\av{\tau}_{\vec p}$, the expression is minimized for a distribution with the largest possible second moment. This is described by a distribution only on the boundary of the parameter space. 
\begin{align}
    \av{\tau}_{\vec p}&=q\sin(x)+(1-q)x\\
    \av{\tau^2}_{\vec p}&=q\sin^2(x)+(1-q)x^2
\end{align}
for $q\in[0,1]$.
\begin{align}
    \left(\frac{\Omega^2_{\mathrm{S}}}{\Omega^2_{\mathrm{UB}}}\right)
    &\geq\min_{\alpha>0}\max_{x(\alpha)\in[0,\pi/2]} \min_{q\in [0,1]}\frac{(q\sin(x)+(1-q)x)^2(1+\alpha)}{q\sin^2(x)+(1-q)x^2+\alpha}
\end{align}
if this expression is minimized for $q$, one obtains
\begin{align}\label{eq:op_ax}
    \left(\frac{\Omega^2_{\mathrm{S}}}{\Omega^2_{\mathrm{UB}}}\right)^*(\alpha,x)
    \geq\begin{cases}
    4(1+\alpha)\frac{x\sin(x)-\alpha}{(x+\sin(x))^2}\qquad 2\alpha\leq \sin(x)x-\sin^2(x)\\
    \frac{\sin^2(x)(1+\alpha)}{a+\sin^2(x)}\qquad\qquad \qquad \mathrm{else}
    \end{cases}
\end{align}
where the latter is always at least 1 for $x=\pi/2$. Meaning that for $\alpha\geq \frac{\pi-2}{4}$, $\Omega_\mathrm{S}$ is always at least as large as $\Omega_{\mathrm{UB}}$. For $\alpha< \frac{\pi-2}{4}$, we find a lower bound by choosing $x(\alpha)= 2\alpha+1$, where $x\in[1,\frac{\pi}{2}]$. This returns
\begin{align}
    \left(\frac{\Omega^2_{\mathrm{S}}}{\Omega^2_{\mathrm{UB}}}\right)^*
    \geq \min_{\alpha\in [0,\frac{\pi-2}{4}]}4(1+\alpha)\frac{(2\alpha+1)\sin(2\alpha+1)-\alpha}{(2\alpha+1+\sin(2\alpha+1))^2}\geq 0.984\,,
\end{align}
which means $\Omega_{\mathrm{S}} \geq 0.99\Omega_{\mathrm{UB}}$ showing \cref{thm:single_is_good}.
\qed 

To proof c\cref{cor:single_is_good} we use \cref{eq:op_ax}, but keep a fixed $x=\pi/2$. This corresponds to a simple central differences strategy at $x=\frac{\pi}{2\nu}$ in the original case. Here we get 
\begin{align}
    \left(\frac{\Omega^2_{\mathrm{S}}}{\Omega^2_{\mathrm{UB}}}\right)\geq 4(1+\alpha)\frac{\pi/2-\alpha}{(\pi/2+1)^2}\geq 2\frac{\pi}{(\pi/2+1)^2}\geq 0.95
\end{align}
regardless of the underlying distribution ($\mathcal{D}_{\theta},\sigma$).\qed
\subsection{Barren plateau and 2-design calculations}
\label{ap:2design}
To find estimates of the size of the Fourier coefficients, we need to estimate $\edistheta{c_{ij}c^*_{kl}}$ given by
\begin{align}
    \edistheta{c_{ij}c^*_{kl}}=\int \bra{\Psi}U^\dagger P_i V^\dagger O V P_j U \ket{\Psi}\bra{\Psi}U^\dagger P_l V^\dagger O V P_k U \ket{\Psi}\mathrm{d}V \mathrm{d}U\,.
\end{align}
Here $U$ describes the ensemble of unitaries that are applied in the \ac{VQA} before the layer of interest, $V$ the unitaries after the layers, but before the measurement. We assume that both $U$ and $V$ describe $2$-designs. This is useful because it allows us to use the identity for Haar random unitaries
\begin{equation}\label{eq:ABtwirl}
    \int U^\dagger A U\rho U^\dagger B U\mathrm{d}U=\frac{\1\Tr(\rho)}{d}\left(\frac {d\Tr(AB)}{d^2-1}-\frac{\Tr(A)\Tr(B)}{d^2-1}\right)+\rho\left(\frac{\Tr(A)\Tr(B)}{d^2-1}-\frac{\Tr(AB)}{d(d^2-1)}\right) \, .
\end{equation}
where $d$ is the Hilbert space dimension.
We set  
\begin{align}
    \rho_{ij}=P_iU\ket{\Psi}\bra{\Psi}U^\dagger P_j
\end{align} 
and use the identity \cref{eq:ABtwirl} to obtain 
\begin{align}
    \edistheta{c_{ij}c^*_{kl}}    &=\int \Tr[V^\dagger O V \rho_{jl} V^\dagger O V\, \rho_{ki}]\, \mathrm{d}V\mathrm{d}U
    \\
    &=\frac{1}{d^2 - 1}\int \Tr\biggl[
        \frac{\1 \Tr[\rho_{jl}]}{d}\left(d \Tr[O^2]-\Tr[O]^2\right)\, \rho_{ki}
        + \rho_{jl} \left(\Tr[O]^2 - \Tr[O^2]/d\right)\, \rho_{ki}
    \biggr]\, \mathrm{d}U 
    \\
    &=\frac{d\Tr[O^2]-\Tr[O]^2}{d(d^2-1)}\int \Tr[\rho_{ki}]\Tr[\rho_{jl}]\, \mathrm{d}U
    +\frac{\Tr[O]^2-\Tr[O^2]/d}{d^2-1}\int \Tr[\rho_{ki}\rho_{jl}]\, \mathrm{d}U
\end{align}
For the relevant terms, i.e.\ the ones with $i\neq j$, it follows that $\Tr[\rho_{ki}\rho_{jl}]=0$. 
For the first term to not vanish, we require $i=k$ and $j=l$.
With $\Tr[\ketbra \Psi\Psi]=1$ and $\Tr[P_iP_j]=0$, 
\begin{align}
    \int \Tr[\rho_{ii}]\Tr[\rho_{jj}]\, \mathrm{d}U
    &=
    \int \bra{\Psi}U^\dagger P_iU\ket{\Psi}\bra{\Psi}U^\dagger P_jU\ket{\Psi}\, \mathrm{d}U\\
    &=\frac{1}{d}\left(\frac {d\Tr[P_iP_j]}{d^2-1}-\frac{\Tr[P_i]\Tr[P_j]}{d^2-1}\right)
    +
    \left(\frac{\Tr[P_i]\Tr[P_j]}{d^2-1}-\frac{\Tr[P_iP_j]}{d(d^2-1)}\right)\\
    &=\frac{\Tr[P_i]\Tr[P_j]}{d^2-1}\left(1-\frac{1}{d}\right)\\
    &=\frac{\Tr[P_i]\Tr[P_j]}{d(d+1)}\, .
\end{align} 
This leads to
\begin{align}
    \av{|c_{ij}|^2}_{\mathcal{D}}
    &=
    \frac{\Tr[P_i]\Tr[P_j]}{d(d+1)}\left(\frac {d\Tr[O^2]-\Tr[O]^2}{d(d^2-1)}\right)
    \\
    &=\frac1d\frac{\Tr[P_i]\Tr[P_j]}{d^2}\left(\Tr[O^2/d]-\Tr[O/d]^2\right)\frac{d^3}{(d+1)(d^2-1)}\\
    &= \frac{\xi_d}d\frac{\Tr[P_i]\Tr[P_j]}{d^2}\left(\Tr[O^2/d]-\Tr[O/d]^2\right)\\
    &= \frac{\xi_d}d\Tr[P_i/d]\Tr[P_j/d]\sigma_O^2\\
    \end{align}
    with $\xi_d\coloneqq\frac{d^3}{(d+1)(d^2-1)} $, $\sigma_O^2\coloneqq \Tr[O^2/d]-\Tr[O/d]^2$ being the expected variance with respect to the maximally mixed state and $\Tr[P_i]$ is the multiplicity of the eigenvalue $\lambda_i$. The final estimate is therefore
    \begin{align}
    \sigakt&=\sum_{i\geq j:\mu_k=\lambda_i-\lambda_j} \av{|c_{ij}|^2}_{\mathcal{D}}\\
    &=\xi_d\frac{\sigma_O^2}d\sum_{i\geq j:\mu_k=\lambda_i-\lambda_j}\Tr[P_i/d]\Tr[P_j/d]
\end{align}
 
Similarly, one can obtain a shot noise estimate when measuring in the Hamiltonian eigenbasis by the one design property
\begin{align}
    \sigma^2&=\int \bra{\Psi}U^\dagger O^2 U\ket{\Psi}\mathrm{d} U-\left(\int \bra{\Psi}U^\dagger O U\ket{\Psi}\mathrm{d} U\right)^2\\
    &=\Tr[O^2/d]-\Tr[O/d]^2
    =\sigma_O^2\,.
\end{align}
We note that when $O$ is not directly measured in its eigenbasis, the real shot noise variance might be significantly large, as typically multiple different measurement setting are required.
\subsubsection{Twirls of linear maps on operators with unitary 2-design}
\label{sec:two-designs}
Given a linear map $M$ on the vector space of operators and a probability distribution on the unitary group one can define the \emph{twirl} of $M$ as $T$ where
\begin{equation}
    T(X) \coloneqq \EE[U^\dagger M(U X U^\dagger)\,U]\, .
\end{equation}
The distribution of unitaries is called a \emph{unitary 2-design} if the expectation value above yields the same as the similar one for the Haar measure (see, e.g.\ the tutorial \cite{Kliesch2020TheoryOfQuantum} for details). 
In this case, $T$ satisfies the invariance condition $T(X) = U^\dagger T(UXU^\dagger)\, U$ for all operators $X$ and the expectation value can be evaluated as in the following lemma, 

\begin{lemma}[Twirl of maps on operators {\cite[Appendix]{EmeAliZyc05}}]
Let $T: \L(\CC^d) \to \L(\CC^d)$ be a linear map that satisfies the invariance 
$T(X) = U^\dagger T(UXU^\dagger) \,U $ for all $X\in \L(\CC^d)$ and unitaries $U\in \U(d)$. 
Then 
\begin{equation}\label{eq:invariantT}
    T(X) 
    = 
    \frac{\Tr[T(\1)] - \Tr[T]/d}{d^2-1} \, \Tr[X] \, \1
        + \frac{\Tr[T] - \Tr[T(\1)]/d}{d^2-1} \, X \, .
\end{equation}
In particular, 
\begin{equation}\label{eq:ABtwirl2}
    \int U^\dagger A U\rho U^\dagger B U\mathrm{d}U=\frac{\1\Tr(\rho)}{d}\left(\frac {d\Tr(AB)}{d^2-1}-\frac{\Tr(A)\Tr(B)}{d^2-1}\right)+\rho\left(\frac{\Tr(A)\Tr(B)}{d^2-1}-\frac{\Tr(AB)}{d(d^2-1)}\right) \, .
\end{equation}
\end{lemma}

\begin{proof}
We denote the identity map by $\id$ and the completely dephasing channel by $R(X) \coloneqq \Tr[X]\, \1/d$. 
Due to a standard argument \cite[Appendix]{EmeAliZyc05} relying on Schur's lemma one can write 
\begin{equation}\label{eq:T:LinComb}
    T = a R + b\, \id
\end{equation}
for some coefficients $a,b\in \CC$. 
Taking the trace of this equation and of $T(\1) = a R(\1) + b\, \id(\1)$ results in 
\begin{align}
    \Tr[T] &= a + b d^2\, , \\
    \Tr[T(\1)] &= ad + bd\, .
\end{align}
Solving for $a$ and $b$ and inserting these coefficients in the ansatz \cref{eq:T:LinComb} yields the statement \cref{eq:invariantT}. 

Next, we take $T$ to be the RHS of Eq.~\cref{eq:ABtwirl2}. 
Then $T$ satisfies the lemma's invariance condition and one can show that $\Tr[T]=\Tr[A]\Tr[B]$ and $\Tr[T(\1)] = \Tr[AB]$, which results in \cref{eq:ABtwirl2}. 

See also \cite[Theorem~51]{Kliesch2020TheoryOfQuantum} for more explicit calculations relying on the second moment operator of a unitary $2$-design. 
\end{proof}

\subsection{Analytical estimates for the Fourier coefficients of the QAOA Hamiltonian}
\label{ap:a_dependence}
As is mentioned in the main text, to obtain quantitative results for correlation matrix of the \ac{QAOA}, we are considering the expectation value over all graph instances
\begin{align}
    \edisboth{\argdot}\rightarrow \av{\argdot}_{s,\vec\theta,G}\,.
\end{align}
the ensemble is randomly chosen graphs with $N$ vertices and $M$ edges.
\subsubsection{Case for \texorpdfstring{$L\rightarrow\infty$}{large L}}
\label{ap:linf_est}
For $L\rightarrow\infty$ using the derived quantity for the 2-design \cref{eq:sigakt_2design}, we need to calculate the expression
\begin{align}
    \sigakt&=\edisG{\frac{\sigma_O^2}d\sum_{i\geq j:\mu_k=\lambda_i-\lambda_j}\Tr[P_i/d]\Tr[P_j/d]}
\end{align}
over all graph instances.
Since $O=H_c = \frac{1}{2}\sum_{(i,j) \in E} \left(Z_iZ_j - \mathds{1}\right)$, It follows with $\Tr(Z_iZ_j)=0$ that $\Tr(O/d)=-M/2$. Equally, by only counting the terms which are proportional to the identity, we get $\Tr(O^2/d)=\frac{M^2}{4}+\frac{M}{4}$, meaning $\sigma_O^2=\frac{M}{4}$. This quantity is independent of the particular graph taken from the distribution.

We note that the \ac{VQA} is invariant under the parity symmetry in the $X$-basis 
\begin{align}
    \Pi_x=\sigma_x\otimes \cdots\otimes \sigma_x
\end{align}
the state only lives in the even number of ones sector, which also halves the Hilbert space dimension $d=2^{N-1}$.

To find $\left.\edistheta{a_{k}^{2}}\right|_b$, we first note that $H_b$ is independent of the particular graph instance chosen.
The eigenvalues of $H_b$ describe spin sectors with $\lambda_i=-\frac{N}{2}+i$ and $\Tr(P^B_i)=\binom{N}{i}$. Introducing the symmetry reduced subspace means that
\begin{align}
    \Tr(P^b_i)=\left\{\begin{matrix}\binom{N}{i}\; , \quad i\; \textrm{even}\\
    0\; , \quad \textrm{else}
    \end{matrix}\right.
\end{align}
which gives the final expression
\begin{align}
    \left.\edistheta{a_k^2}\right|_b&=\frac{M}{4\cdot8^{N-1}}\sum_{j\in \{0,2,\dots N-k\}}\binom{N}{j+k}\binom{N}{j}\\
    &=\frac{M}{4\cdot8^{N-1}}\frac{(2N)!}{(N-k)!(N+k)!}\label{eq:a_b_barren}
\end{align}
For $H_c$, determining the expression requires to take a distribution over the graph instances. 
Qualitatively, it would already be sufficient to calculate $\edisG{\Tr(P^c_i)}$ individually, but for completeness we will derive the full correlation between the two terms.

The relevant expression is
\begin{align}
	 \zeta_k=\sum_{i=0}^M\av{\Tr(P^c_k) \Tr(P^c_{k+i})}_{G}/d^2\, ,
\end{align}
which can be calculated by solving a combinatorial problem. For this, we reformulate the question as asking what the likely hood is that for two different partitions the fist cuts $k$ edges and the second $k+j$ edges. This is because by design $\Tr(P^c_k)/d$, describes the likelyhood of a random partition having $k$ edges cut.
To find this quantity, we classify the vertices to belong to any of the four sectors $\{(0,0),(0,1),(1,0),(1,1)\}$ indicating on which side of the cuts the vertex is placed ($(0,0)$ mean it is on the first side for both cuts) The number of vertices in each sectors is $\vec s=(s_{00},s_{01},s_{10},s_{11})$ with the probability for a particular $\vec s$ being
\begin{align}\label{eq:p(s)}
   p(\vec s)=\binom{n}{s_{00},s_{01},s_{10},s_{11}} \frac{1}{4^n}\,
\end{align}
as the process of chosing vertex position is described by a multinomial distribution where each sector has probability $p=1/4$.
Next, we sample the edges: From $|\Gamma|=\binom{n}{2}$ possible edges, $E_1=s_{00}s_{10}+s_{01}s_{11}$ edges are cut only by the first partition and $E_2=s_{00}s_{01}+s_{10}s_{11}$ cut only the second 
We do not consider the cases where the edges is cut by neither or both partitions as this will not affect the overall result.  As the particular edges are chosen randomly from $\Gamma$, the process is described by a multi-hypergeometric distribution 
\begin{align}
    \zeta_k=\sum_{\vec s\in [n]^4,i\in [M]} p(\vec s) \frac{\binom {E_1(\vec s)}{i} \binom {E_2(\vec s)}{i+k}\binom{|\Gamma|-E_1(\vec s)-E_2(\vec s)}{M-2i-k}}{\binom{|\Gamma|}{M}}
\end{align}

which gives the final result for $H_c$
\begin{align}\label{eq:a_c_barren}
    \left.\edisG{\edistheta{a_{k}^{2}}}\right|_c&=\frac{M}{2^{N-3}}\zeta_k\,.
\end{align}
\subsubsection{Case for one layer \texorpdfstring{$L=1$}{L=1}}
In the following section we derive a single layer ($L=1$) estimate, by relating the underlying problem to a graph problem. We note that rigorous analysis performed here is excessive, because for practical purposes simply sampling the derived graph problem for the particular graph instance will already give a reliable estimate. However, for completeness, we derive exact analytic results for the \ac{QAOA}.
\label{ap:l1_est}
We begin by defining three observables
\begin{align}
    O_{zz}&=
    \frac12
    \sum_{ij\in E} Z_iZ_j\\
    O_{yy}&=
    \frac12
    \sum_{ij\in E} Y_iY_j\\
    O_{yz}&=
    \frac12
    \sum_{ij\in E} Y_iZ_j+Z_iY_j\,.
\end{align}
Note that the derivative is not affected by replacing $H_c\rightarrow O_{zz}$ as this only describes a constant offset. 
The cost function for the first layer can be subdivided
\begin{align}
    F(\gamma,\beta)
    &=\frac 12 (1+\cos(2\beta))C_{zz}(\gamma)+\frac 12 (1-\cos(2\beta))C_{yy}(\gamma)+\frac 12\sin(2\beta)C_{yz}(\gamma)
\end{align}
where
\begin{align}
    C_\rho(\gamma)&= \sum_{ij} \bra{\Psi_i}O_{\rho}\ket{\Psi_j}\e^{-\i \gamma (\lambda_i-\lambda_j)}\\
    &= \sum_{k=0}^M \left(\sum_{i=0}^{M-k}\bra{\Psi_i}O_{\rho}\ket{\Psi_{i+k}}\e^{\i k\gamma}\right)\\
    &\eqqcolon\sum_{k=0}^M C_\rho^k(\gamma)
\end{align}
with $\rho \in \{zz,yy,yz\}$ and $\ket{\Psi_j}=P_j^c\ket{+}$.
We can quickly check that
\begin{align}
    C_{zz}(\gamma)&= \sum_{ij} \bra{\Psi_i}O_{zz}\ket{\Psi_j}\e^{-\i \gamma (\lambda_i-\lambda_j)}\\
     &=\sum_{i} \bra{\Psi_i}O_{zz}\ket{\Psi_i}
     =\sum_{i} \bra{+}P_iO_{zz}P_i\ket{+}
     =\bra{+}O_{zz}\ket{+}
    =0
\end{align}
which
means the expression simplyifies to
\begin{align}
    F(\gamma,\beta)
    &=\frac 12 (1-\cos(2\beta))C_{yy}(\gamma)+\frac 12\sin(2\beta)C_{yz}(\gamma)
\end{align}
The evolution of $H_b$ only has one non-vansihing frequency $\edis{a_2^2}$. By integrating over $\beta$ and $\gamma$ w.r.t. the correct frequency we get
\begin{align}
    \left.\edis{a_2^2}\right|_b=\frac{\av{|C_{yy}|^2}_{\gamma}+\av{|C_{yz}|^2}_\gamma}{8}=\sum_{k=0}^{M}\frac{\av{|C^k_{yy}|^2}_{\gamma}+\av{|C^k_{yz}|^2}_\gamma}{8}
\end{align}
and
\begin{align}
    \left.\edis{a_k^2}\right|_c=\sum_{k=0}^{M}\frac{3\av{|C^k_{yy}|^2}_{\gamma}+\av{|C^k_{yz}|^2}_\gamma}{8}
\end{align}
where 
\begin{align}\label{eq:fouriersq}
    \av{|C^k_{\rho}|^2}_{\gamma}=\sum_{i=0}^M\bra{\Psi_i}O_{\rho}\ket{\Psi_{i+k}}\times \sum_{j=0}^M\bra{\Psi_{j+k}}O_{\rho}\ket{\Psi_{j}}\,.
\end{align}
The remainder of this section shows how~\cref{eq:fouriersq} can be interpreted as a graph problem and how to calculate the expression for general graphs with $N$ vertices and $M$ edges.
Similar to in the last section, computational basis describes a particular bi-partition of vertices with every qubit describing one vertex.
with each block describing different partitions, the vertices in each block the number of vertices in each are labeled again by $\vec s=(s_{00},s_{01},s_{10},s_{11})$ with the same probability as in \cref{eq:p(s)}. $O_\rho$ describes a certain graph operation which will be performed on both partitions. Only if in both partitions, the number of edges cut changes by $k$, will this partition pair contribute to $\av{|C^k_{\rho}|^2}_{\gamma}$.

For the action of $O_\rho$, with $\rho \in \{yy,yz \}$, we use $Y=-iXZ$ to rewrite $YY= -XX\times ZZ$ and $ZY+YZ=-i(\1X+\1X)\times ZZ$. Effectively $X$ corresponds to moving a vertex to the other side of the respective bi-partition and $ZZ$ accumulates a sign if the edge vertices are on opposing sides.
The overall sign is therefore determined by the rule
\begin{align}
    \mathrm{sign}(E_1,E_2)=\begin{cases}
    +& \textrm{ Both edges cut their partition or both do not.}\\
    -& \textrm{ One edges cuts its partition the other does not.}\
    \end{cases}
\end{align}
\newcommand{\w}{\omega}
Each operator $O_\rho$ sums over all edges which gives $M^2$ different edge pairs to consider. There are $\w_{\mathrm{same}}=M$ cases where the two considered are the same edge in the graph. When this is not the case, we will separate into additional cases depending on if the two edges share a vertex or not.
\begin{enumerate}
    \item The edges are the same edge: $\w_{\mathrm{same}}=M$
    \item The edges have a common vertex: $\w_{\mathrm{con}}=(M^2-M)\frac{2(N-2)}{\binom{M}{2}-1}$, where $2(N-2)$ is the number of possible edges that connect to the first edge.
    \item The have no overlapping vertex: $\w_{\mathrm{sep}}=M^2-\w_{\mathrm{same}}-\w_{\mathrm{con}}$.
\end{enumerate}
We then procced to draw the corresponding vertices depending on each of the three cases, which means drawing $2,3$ or $4$ vertices respectively. The probability of choosing an individual vertex from each set is 
\begin{align}
    q_{\vec s}(v_1=(i,j))&=\frac{s_{ij}}{\|\vec {s}\|_1}\\
    q_{\tilde{\vec s}}(v_l=(i,j))&=\frac{\tilde s_{ij}}{\|\tilde{\vec s}\|_1}\,\quad \textrm{where} \quad \tilde{\vec s}=\vec s_{\backslash \{v_1,\dots,v_{l-1}\}}\,.
\end{align}
$s_{\backslash V}$ is the sector distribution of vertices without the vertices in the set $V$. Additionally, if $\rho=yz$, we also need to decide on which the $Y$ operation acts. For this we place each vertex into the edge sets $\vec E=(E_1^Y,E_1^Z,E_2^Y,E_2^Z)\in S_{\rho,\xi}$ which encapsulates which operation operation is performed on them according to the correct $(\rho,\xi)$. If a vertex is shared, it enters twice in this edge set. We use the smaller indices for the first edge and the shared edges are also starting from $v_1$. $S_{\rho,\xi}$ describes all legal operations, for instance
    \begin{align}
        S_{yy,\mathrm{same}}&=\{(\{v_1,v_2\},\emptyset,\{v_1,v_2\},\emptyset)\}\\
    S_{yz,\mathrm{con}}&=\{ (\{v_1\},\{v_2\},\{v_1\},\{v_3\}),(\{v_2\},\{v_1\},\{v_1\},\{v_3\}),(\{v_1\},\{v_2\},\{v_3\},\{v_1\}),(\{v_2\},\{v_1\},\{v_3\},\{v_1\})\}
    \end{align}

Now, we are able to proceed to draw the remaining edges of the graph. Notably, they only contribute to $\av{|C^k_{\rho}|^2}_{\gamma}$ if they are connected to the drawn vertices with an $X$ operation. 
To summerize, the steps are:
\begin{enumerate}
    \item Select $\rho\in\{yy,yz\}$.
    \item Select a bipartitions $\vec s$ with the probability $p(\vec s)$ as defined in~\cref{eq:p(s)}.
    \item Select the relationship of the edges to each other $\xi\in(\mathrm{same,con,sep})$.
    \item Select $2-4$ vertices from $\vec s$ according to the particular case $\xi$.
    \item Select a vertex distribution according to the correct edge set $\vec E=(E_1^Y,E_1^Z,E_2^Y,E_2^Z)\in S_{\rho,\xi}$.
    \item Draw the remaining edges of the graph from a hyper-geometric distribution to calculate the effect on the vertices cut.
\end{enumerate}
 The final equation is given as
 \begin{align}
    \edis{|C_\rho^k|^2}=\sum_{\substack{\vec s \in \{0,\dots,N\}^4|\sum s_i=N \\\xi\in \{\textrm{same,con,sep}\}\\
    \vec v\in \{0,1\}^{2\times n_\xi} \\
    \vec E \in S_{\rho,\xi}}}
     p(\vec s)\w_\xi q_{\vec {s}}(v_1)\cdots q_{\vec {\tilde s}}(v_{n_\xi})\mathrm{sign}(E_1,E_2) p(k|\vec s,\vec E)\,
\end{align}
where $p(k|\vec s,\vec E)$ is the probability of the changes on both bipartitions being $k$.

There are guaranteed changes which we label with $\vec F(\vec E)$ which arises from the effect of the already drawn edges onto each other. This only relevant if the edges share vertices on which one $X$ operation operates, meaning $\xi=\mathrm{con}$ and if $\rho=yz$ also $\xi=\mathrm{same}$. The sign is determined by if the other edge is cut in the bipartition or not.
\begin{align}
    F_1(\vec E)=\delta_{|E^Y_1\cap E_2|,1}(-1)^{\sum_{w\in E_2} w_1}\\
    F_2(\vec E)=\delta_{|E_1\cap E^Y_2|,1}(-1)^{\sum_{w\in E_1} w_2}
\end{align}
where $w_i$ refers to the labeling for the $i$-th partition.
The edges themselves are drawn from a hyper-geometric distribution.
From the full edge set $\Gamma$ with $|\Gamma|=\binom{N}{2}$ after the one/two edges are chosen there remains $\Gamma_r$ edges, with $|\Gamma_r|=\binom{N}{2} -1$ for $\xi=\mathrm{same}$ and $|\Gamma_r|=\binom{N}{2} -2$ for the others. Similarly the number of edges to select are $M_r=M-1$ and $M_r=M-2$.
Any of these edges fits into one of $9$ categories, either not affecting, increasing or decreasing the value of the cut after performing the operation for each of the partitions.
For this we define $K_{\alpha,\beta}$ with $(\alpha,\beta)\in\{-1,0,1\}^2$.  The number of edges drawn from each category is labeled $X_{\alpha,\beta}$.  
We can determine $K$ with the following rules
\begin{align}
    K_{1-2m,0}&=\sum_{i\in\{0,1\},v\in E_1^Y\setminus E_2^Y} \tilde s_{v_1\oplus m,i}\qquad &\vec{\tilde s}&=\vec s_{\backslash \{E_1^Y\cup E_2^Y\cup E_1^Z\}}\\
    K_{0,1-2m}&=\sum_{i\in\{0,1\}v\in E_2^Y\setminus E_1^Y} \tilde s_{i, v_2\oplus m}\qquad &\vec{\tilde s}&=\vec s_{\backslash \{E_1^Y\cup E_2^Y\cup E_2^Z\}}\\
    K_{1-2m,1-2l}&=\sum_{v\in E_1^Y\cap E_2^Y}\tilde s_{v \oplus(m,l)}+\sum_{v\in E_1^Y,w\in E_2^Y} \delta_{v_1\oplus w_2,(m,l)}
    \qquad &\vec{\tilde s}&=\vec s_{\backslash \{E_1^Y\cup E_2^Y\cup E_1^Z\cup E_2^Z\}}\\
    K_{0,0}&=|\Gamma_r|-\sum_{\substack{\alpha,\beta\in \{-1,0,1\}^2\\ (\alpha,\beta)\neq (0,0)}}K_{\alpha,\beta}
\end{align}
Here $E_1\cap E_2$ refers to vertices that are shared by the two edges. 
The corresponding probability distribution is then given by
\begin{align}
    g(k|\vec s,\vec E)=g(k|\vec K(\vec s,\vec E),\vec F(\vec E))&=\frac{1}{\binom{|\Gamma_{r}|}{M_{r}}}\sum_{\substack{\vec X \in [\vec K]\\\sum_{ij} X_{ij}=M_r\\\sum_{\beta} (X_{+1,\beta}-X_{-1,\beta})+F_1=k\\
    \sum_{\alpha} (X_{\alpha,+1}-X_{\alpha,-1})+F_2=k
    }} \prod_{(\alpha,\beta)\in \{-1,0,1\}^2} \binom{K_{\alpha,\beta}}{X_{\alpha,\beta}}
\end{align}
With this we have all the ingredients to calculate $\left.\edisG{\edistheta{a_{k}^{2}}}\right|_b$ and $\left.\edisG{\edistheta{a_{k}^{2}}}\right|_c$ the result of which is plotted in \cref{fig:a_dependence}


\end{document}